\documentclass[11pt,twoside]{article}
\usepackage[letterpaper,margin=1in]{geometry}
\usepackage[utf8x]{inputenc}
\usepackage{amssymb}
\usepackage{graphicx}
\usepackage{jeffe}
\usepackage{microtype}
\usepackage{soul}

\usepackage[charter]{mathdesign}

\graphicspath{{Fig/}}
\newtheorem{lemma}{Lemma}[section]
\newtheorem{corollary}[lemma]{Corollary}

\newtheorem{theorem}[lemma]{Theorem}

\newcommand{\create}{\ensuremath{\textsc{Create}}}
\newcommand{\join}{\ensuremath{\textsc{Link}}}
\newcommand{\cut}{\ensuremath{\textsc{Cut}}}
\newcommand{\getnodevalue}{\ensuremath{\textsc{GetNodeValue}}}
\newcommand{\addsubtree}{\ensuremath{\textsc{AddSubtree}}}

\newcommand{\minpath}{\ensuremath{\textsc{MinPath}}}
\newcommand{\addpath}{\ensuremath{\textsc{AddPath}}}
\newcommand{\getdartvalue}{\ensuremath{\textsc{GetDartValue}}}
\newcommand{\junction}{\ensuremath{\textsc{Junction}}}

\newcommand{\lca}{\ensuremath{\textsc{LCA}}}
\newcommand{\evert}{\ensuremath{\textsc{Evert}}}

\newcommand{\GroveJoin}{\ensuremath{\textsc{GroveLink}}}
\newcommand{\GroveCut}{\ensuremath{\textsc{GroveCut}}}

\def\dart#1#2{#1\mathord\shortrightarrow#2}
\def\head{\mathit{head}}
\def\tail{\mathit{tail}}
\def\lsh{\mathit{left}}
\def\rsh{\mathit{right}}
\def\val{\mathit{val}}
\def\dist{\mathit{dist}}
\def\pred{\mathit{pred}}
\def\slack{\mathit{slack}}
\def\excess{\mathit{excess}}

\def\dvec#1{\vec{#1}\,^*}	   

\def\hw{\widehat{w}}
\def\concat{\cdot}
\def\Gun{G_{\rm un}}
\title{Multiple-Source Shortest Paths in Embedded Graphs%
\thanks{A preliminary version of this work, without the third author's contributions, was presented at the 18th Annual ACM-SIAM Symposium on Discrete Algorithms~\cite{cc-07}.}}

\author{
  Sergio Cabello%
	\thanks{Department of Mathematics, IMFM, and Department of Mathematics, FMF, University of Ljubljana, Slovenia, \protect\url{sergio.cabello@fmf.uni-lj.si}. Research partially supported by the European Community Sixth Framework Programme under a Marie Curie Intra-European Fellowship, and by the Slovenian Research Agency, projects P1-0297, J1-7218 and J1-4106. 
}
  \and
  Erin W. Chambers%
  	\thanks{
Department of Mathematics and Computer Science, Saint Louis University, \protect\url{echambe5@slu.edu}.  Portions of this work were done while the author was affiliated with the University of Illinois, Urbana-Champaign.  Research partially supported by an NSF Graduate Research Fellowship and by NSF grants DMS-0528086 and CCF-1054779.}
  \and
  Jeff Erickson%
  	\thanks{Department of Computer Science, University of Illinois, Urbana-Champaign, \protect\url{jeffe@cs.uiuc.edu}.  Research partially supported by NSF grants DMS-0528086 and CCF-0915519.}
}

\begin{document}
\begin{titlepage}
\maketitle

\begin{abstract}
Let $G$ be a directed graph with $n$ vertices and non-negative weights in its directed edges, embedded on a surface of genus~$g$, and let $f$ be an arbitrary face of $G$.  We describe a randomized algorithm to preprocess the graph in $O(gn \log n)$ time with high probability, so that the shortest-path distance from any vertex on the boundary of~$f$ to any other vertex in~$G$ can be retrieved in $O(\log n)$ time.  Our result directly generalizes the $O(n\log n)$-time algorithm of Klein [Multiple-source shortest paths in planar graphs. In \emph{Proc.\ 16th Ann.\ ACM-SIAM Symp.\ Discrete Algorithms}, 2005] for multiple-source shortest paths in planar graphs.  Intuitively, our preprocessing algorithm maintains a shortest-path tree as its source point moves \emph{continuously} around the boundary of $f$.  As an application of our algorithm, we describe algorithms to compute a shortest non-contractible or non-separating cycle in embedded, undirected graphs in $O(g^2 n\log n)$ time with high probability.  Our high-probability time bounds hold in the worst-case for generic edge weights, or with an additional $O(\log n)$ factor for arbitrary edge weights.

\bigskip\noindent
\textbf{Keywords:}
computational topology, topological graph theory, parametric shortest paths, dynamic data structures
\end{abstract}

\thispagestyle{empty}
\setcounter{page}{0}
\end{titlepage}

\section{Introduction}
\label{sec:intro}

Let $G$ be a directed graph with $n$ vertices and non-negative weights in the directed edges, embedded on a surface of genus $g$, and let $f$ be an arbitrary face of $G$.  In this paper, we describe an algorithm to preprocess the graph in $O(g n \log n)$ time and space, so that later the shortest-path distance from any vertex on the boundary of~$f$ to any other vertex in $G$ can be retrieved in $O(\log n)$ time.  Our preprocessing algorithm constructs an implicit representation of all shortest-path trees rooted at vertices of $f$.  Storing all these shortest-path trees explicitly would require $\Theta(n^2)$ space in the worst case. 

Euler's formula implies that any graph with $n$ vertices embedded in a surface of genus $g$ has $O(g+n)$ edges.  Our problem can be solved by solving the all-pairs shortest-path problem in $G$ in $O(gn + n^2\log n)$ time using Dijkstra's algorithm \cite{d-ntpcg-59} with Fibonacci heaps~\cite{ft-fhtui-87}.  When $g=O(1)$, the running time can be reduced to $O(n^2)$ using the faster shortest-path algorithm of Henzinger \etal\ \cite{hkrs-fspap-97}; see also~\cite{tm-spltm-09}.  Our algorithm improves both of these running times if and only if $g = o(n)$.  Thus, for the rest of the paper, we assume $g=o(n)$, which implies that $G$ has $O(n)$ directed edges.

Our results generalize (and were inspired by) an algorithm of Klein \cite{k-msspp-05} that solves the corresponding multiple-source shortest-path problem in \emph{planar} graphs in $O(n\log n)$ time.  Other noteworthy results on multiple-source shortest paths for planar graphs include Frederickson's all-pairs shortest-path representation \cite{f-faspp-87}, Lipton and Tarjan's planar separator theorem \cite{lt-stpg-79}, and Schmidt's $O(n \log n)$ algorithm that supports distance queries for specific subsets of vertices on a grid \cite{284959}.

Let $v_0, v_1, \dots, v_k$ be the sequence of vertices around the specified face $f$.  Klein's algorithm begins with a shortest-path tree $T$ rooted at $v_0$ and then proceeds in $k$ phases; at the end of the $i$th phase, $T$ is the shortest-path tree rooted at $v_i$.  At the beginning of the $i$th phase, the algorithm deletes the directed edge in $T$ leading into $v_i$, adds the directed edge from $v_i$ to $v_{i-1}$, and declares $v_i$ to be the new root.  The algorithm then transforms $T$ back into a shortest-path tree through a series of \emph{pivot} operations; each pivot changes the predecessor of some vertex, replacing its previous incoming edge with a new directed edge, just as in the classical shortest-path algorithms of Dijkstra \cite{d-ntpcg-59} and Bellman and Ford~\cite{f-nft-56}.  Specifically, at each iteration, Klein's algorithm pivots the \emph{leafmost unrelaxed directed edge} into~$T$.  This characterization exploits the classical observation, originally by von Staudt \cite{vs-bg-1847}, that for any spanning tree $T$ of any planar graph~$G$, the edges not in $T$ comprise a spanning tree of the dual graph $G^*$.  Klein describes how to find the leafmost unrelaxed directed edge and pivot it into $T$ in $O(\log n)$ time, using a dynamic tree data structure \cite{st-dsdt-83, tw-satt-05}.  He also proves that each directed edge pivots into $T$ at most a constant number of times; the $O(n \log n)$ running time follows immediately.  To allow for future shortest-path queries, Klein maintains $T$ in a \emph{persistent} data structure~\cite{dsst-mdsp-89}.

The main obstacle to extend Klein's algorithm to higher-genus graphs is that the complement of a spanning tree is no longer a dual spanning tree, so there is no `leafmost' unrelaxed directed edge.  Our algorithm follows Klein's basic approach, but uses a different strategy for choosing which directed edges to pivot.  Conceptually, we move the root of $T$ \emph{continuously} around the boundary of $f$ and maintain the correct shortest-path tree at all times.  Even though the source vertex moves continuously, the combinatorial structure of the shortest-path tree changes only at certain critical events, when one or more directed edges pivot into the tree.  Our approach can be seen both as an example of the \emph{kinetic data structure} paradigm proposed by Basch, Guibas, and Hershberger \cite{bgh-dsmd-99, g-kdssar-98} and as a special case of the \emph{parametric shortest-path} problem introduced by Karp and Orlin \cite{ko-pspaa-81, yto-fpspm-91}.

A key property shown in our analysis is that if we move the source vertex along a single edge of $G$, any directed edge of $G$ pivots into the shortest-path tree at most once; thus, the number of pivots is equal to the symmetric difference between the initial and final shortest-path trees.  Moreover, as the source vertex moves around the boundary of $f$, each directed edge pivots into $T$ at most $O(g)$ times, assuming unique shortest paths.  Our algorithm finds and executes each pivot in $O(g\log n)$ time, using a complex data structure composed of $O(g)$ dynamic trees \cite{st-dsdt-83,tw-satt-05}.  These two bounds immediately imply that we can move the source completely around $f$ in $O(g^2n\log n)$ time; a more refined approach improves this time bound to $O(gn\log n)$.  We describe our algorithm in the simpler setting of planar graphs in Section \ref{S:planar} and in full generality in Section \ref{S:surfaces}.

In Section \ref{S:cycles}, as applications of our shortest-path algorithm, we describe algorithms to compute a shortest non-contractible cycle and a shortest non-separating cycle in a combinatorial surface in $O(g^2 n\log n)$ time.  Thomassen developed the first algorithm to compute shortest non-contractible or non-separating cycles, by exploiting the so-called \emph{3-path condition} in $O(n^3)$ time~\cite{t-egnsn-90}; see also Mohar and Thomassen~\cite[Sect.~4.3]{mt-gs-01}.  Erickson and Har-Peled described a faster algorithm that runs in $O(n^2 \log n)$ time \cite{schema}.  Cabello and Mohar~\cite{cm-fsnnc-07} gave an algorithm which runs in time $g^{O(g)}n^{3/2} \log n$, the first algorithm for this problem to have a running time which was bounded by a function of the genus.  Cabello \cite{c-mdpg-12} later improved the running time using separators to $g^{O(g)}n^{4/3}$.  Finally, for \emph{orientable} surfaces Kutz \cite{k-csntc-06} developed an algorithm with a running time of $g^{O(g)} n \log n$ which computed a finite portion of the universal cover and then did standard shortest-path computations in the resulting planar graph.  Our algorithm has a better dependence on the genus and, for nonorientable surfaces of bounded genus, it is the first one to achieve a running time of $O(n\log n)$.

The initial analysis of our algorithms assumes that edge weights are \emph{generic}.  Specifically, we assume that there is always a unique shortest path from the source to each vertex, except at critical events, where exactly one vertex has exactly two shortest paths from the source.   In Section~\ref{S:perturb}, we describe two different perturbation schemes that remove this genericity assumption.  The first is a simple randomized scheme based on the Isolation Lemma of Mulmuley \etal~\cite{mvv-memi-87} that increases the running times of our algorithm at most a constant factor with high probability.  The second is an efficient implementation of lexicographic perturbation \cite{c-odlp-52, hm-apmcp-94} that increases the worst-case running times of our algorithms by a factor of $O(\log n)$.  Klein's planar multiple-source shortest-path algorithm \cite{ek-lafms-13,k-msspp-05} avoids degeneracy issues by maintaining a \emph{rightmost} shortest-path tree; however, it is not clear how to extend his strategy to graphs on higher-genus surfaces.

Since the preliminary version~\cite{cc-07} of this work appeared, our algorithms have been applied by several authors \cite{bateni-steiner-jacm,bdt-09,wobble,homcover,essential,f-fsncs-11,kt-11,kks-11}, and our parametric shortest-path techniques have been applied to maximum flows in planar graphs \cite{parshort} and replacement paths in surface graphs \cite{repath}.  Most recently, Eisenstat and Klein have used our methods to develop $O(n)$-time algorithms for maximum flows and multiple-source shortest paths in unweighted planar graphs \cite{ek-lafms-13}.

\section{Background}
\label{S:background}

\subsection{Surfaces, Embeddings, and Duality}

\paragraph{Surfaces.~}
A \EMPH{surface} (alternatively, a \emph{topological 2-manifold with boundary}) is a Hausdorff topological space in which every point has an open neighborhood homeomorphic to either the plane $\Real^2$ or the closed upper half-plane.  The points that do not have neighborhoods homeomorphic to $\Real^2$ constitute the \EMPH{boundary} of the surface.  A \EMPH{cycle} in a surface $\Sigma$ is (the image of) a continuous map $\gamma\colon S^1\to\Sigma$; a cycle is \EMPH{simple} if this map is injective.  The \EMPH{genus} of a surface $\Sigma$ is the maximum number of simple disjoint cycles $\gamma_1, \gamma_2, \dots, \gamma_g$ such that the subspace $\Sigma\setminus(\gamma_1\cup\cdots\cup\gamma_g)$ is connected.  A surface is \EMPH{orientable} if it contains no subspace homeomorphic to the Möbius band.  We consider only compact, connected surfaces in this paper; up to homeomorphism, such surfaces are uniquely determined by genus, orientability, and the number of boundaries.  When dealing with an orientable surface (which will be the majority of the time), we also fix an orientation on the surface, so that terms like `left', `right', `clockwise', and `counterclockwise' are well-defined.

\paragraph{Curves and Homotopy.}
A \EMPH{path} is (the image of) a continuous map $p\colon [0,1]\to\Sigma$; 
it is \EMPH{simple} if this map is injective. 
The \EMPH{endpoints} of a path $p$ are the points $p(0)$ and $p(1)$.
A \EMPH{loop} is a path $p$ with $p(0)=p(1)$.
An \EMPH{arc} is a path whose endpoints are on the boundary of $\Sigma$.
The term \EMPH{curve} refers to a cycle, path, or arc.
If two paths $p$ and $q$ satisfy $p(1)=q(0)$, their \EMPH{concatenation} \EMPH{$p\concat q$}
is the path defined by $(p\concat q)(t)=p(2t)$, if $t\le 1/2$, and $(p\concat q)(t)=q(2t-1)$, if $t>1/2$.

A \EMPH{homotopy} between two paths $p$ and $q$ is a continuous map
$H: [0,1] \times [0,1] \rightarrow \Sigma$ such that $H(0, \cdot) = p$,
$H(1, \cdot) = q$, $H(\cdot,0) = p(0) = q(0)$, and $H(\cdot,1) = p(1) = q(1)$.
A homotopy between two arcs $\alpha$ and $\beta$ is a continuous map
$H: [0,1] \times [0,1] \rightarrow \Sigma$ such that $H(0, \cdot) = \alpha$,
$H(1, \cdot) = \beta$, $H(\cdot,0)$ is always in the same boundary component, 
and $H(\cdot,1)$ is always in the same boundary component.
A homotopy between two cycles $\gamma$ and $\beta$ is a continuous map
$H:[0,1] \times S^1$ such that $H(0, \cdot) = \gamma$ and $H(1, \cdot) = \beta$.  
Two curves $\beta,\gamma$ are \EMPH{homotopic} when there
is some homotopy between them, and we denote it by \EMPH{$\beta\sim\gamma$}. 
Being homotopic is an equivalence relation.
Up to homotopy, each boundary component of $\Sigma$ admits two parametrizations as simple cycles, one being the reverse of the other.
We say that a cycle is homotopic to a boundary component of $\Sigma$ when it is homotopic to either of these parametrizations. 
If $\gamma$ is a simple cycle homotopic to a boundary $\delta$ of $\Sigma$,
and $\gamma$ and $\delta$ are disjoint,
then $\Sigma \setminus \gamma$ has two connected components, one of them
a topological annulus with $\gamma$ and $\delta$ as boundaries.

A curve is \EMPH{contractible} if it is homotopic to a constant map. In particular, a contractible arc must have its endpoints in the same boundary component $\delta$ and, after contracting $\delta$ to a point, the obtained loop must be contractible.  A simple curve is \EMPH{separating} if $\Sigma \setminus \gamma$ has two connected components. (This concept is related to that of $\mathbb{Z}_2$-homology, but we will not use homology in this paper.)  A simple cycle $\gamma$ is contractible if and only if it bounds a disk, that is, $\Sigma\setminus \gamma$ has two connected components, one of them being a topological disk.

\paragraph{Darts and Embeddings.}
Let $G = (V, \vec{E})$ be a directed graph.  Following Borradaile and Klein \cite{b-epnfc-08, bk-amfdp-09}, we refer to the directed edges in $\vec{E}$ as the \EMPH{darts} of $G$.  (The term \emph{arc} has been used for something else above.)  Each dart connects two vertices, called its \EMPH{tail} and its \EMPH{head}; we will write \EMPH{$\dart{u}{v}$} to denote a dart with tail $u$ and head $v$.  Each dart~$\dart{u}{v}$ has a unique \EMPH{reversal}, defined by swapping its endpoints; thus, the reversal of $\dart{u}{v}$ is $\dart{v}{u}$.  We will assume that whenever $G$ contains a dart $\dart{u}{v}$, then it also contains its reversal $\dart{u}{v}$.  This can be enforced adding any missing darts with a large enough weight, like for example twice the sum of the weights of all original darts.  (In our presentation we need the weights to be finite.) 

Each dart $\dart{u}{v}$ defines an associated (undirected) edge $uv$, which does not carry any orientation.  An edge $uv$ represents the pair of darts $\{ \dart{u}{v}, \dart{v}{u}\}$.  For any directed graph $G = (V, \vec{E})$, we define its associated undirected graph $\Gun=(V,E)$ by replacing each dart with the associated edge.  In most cases, there is no need to distinguish between $G$ and $\Gun$ and we will use $G$ to denote either graph.  

Informally, an \EMPH{embedding} of a directed graph~$G$ on a surface $\Sigma$ is an embedding of its associated undirected graph: vertices are mapped to distinct points and edges are mapped to simple paths on $\Sigma$ that connect the corresponding vertices.  The darts $\dart{u}{v}$ and $\dart{v}{u}$ are embedded using the path of $uv$, but with different orientation.  A \EMPH{face} of an embedding is a maximal connected subset of~$\Sigma$ that does not intersect the image of any edge or vertex.  An embedding is \EMPH{cellular} (or \emph{2-cell} \cite{mt-gs-01}) if every face is an open topological disk; in particular, if the surface has boundary, every boundary cycle must be covered by edges of the graph.  Any cellular embedding can be represented combinatorially by a rotation system and a signature. A \EMPH{rotation system} is a permutation $\pi$ of the darts of~$G$, where $\pi(\vec{e})$ is the dart that appears immediately after $\vec{e}$ in the  circular ordering of darts leaving $\tail(\vec{e})$.  A \EMPH{signature} is a mapping $i\colon E(\Gun)\to\{0,1\}$ assigning a bit to each edge that indicates whether the cyclic ordering at its endpoints are in the same direction or the opposite direction.

When the surface is orientable, the signature is not needed, and we can assume that the rotation system gives a counterclockwise ordering of the darts emanating from each vertex. Every dart in an embedded graph $G$ separates two (possibly equal) faces. For orientable surfaces we may talk of the \EMPH{left shore} and \EMPH{right shore}, respectively denoted \EMPH{$\lsh(\vec{e})$} and \EMPH{$\rsh(\vec{e})$}.  Reversing any dart swaps its shores: $\lsh(\dart{v}{u}) =\rsh(\dart{u}{v})$ and $\rsh(\dart{v}{u}) = \lsh(\dart{u}{v})$.

If $G$ is embedded on an orientable surface of genus $g$,  Euler's formula $\abs{V}-\abs{E}+\abs{F} = 2-2g$ implies that $G$ has at most $3n-6+6g$ edges and at most $2n-4+4g$ faces, with equality if every face of the embedding is a triangle.  If the surface is non-orientable, then Euler's formula $\abs{V}-\abs{E}+\abs{F} = 2-g$ implies that $G$ has at most $3n-6+3g$ edges and at most $2n-4+2g$ faces.  We consider only graphs and surfaces with $g = o(n)$, so in either case the overall complexity of any embedding is $O(n)$.

Our algorithms require an explicit cellular embedding as part of the input.  If no embedding is provided, it is easy to construct an embedding of any given graph on \emph{some} surface, but the genus of this surface may be $\Omega(n)$ even when the minimum possible genus is constant~\cite{ckk-nagg-97}.  Determining whether a given graph $G$ can be embedded on an surface of genus $g$ is {NP}-complete~\cite{t-ggpnp-89}, although there is an embedding algorithm that runs in $g^{O(g)}n\log n$ time \cite{kmr-sltae-08}.  Moreover, for \emph{any} fixed $\e>0$, approximating the genus of an $n$-vertex graph within an additive error of $O(n^{1-\e})$ is {NP}-hard~\cite{ckk-nagg-97}.

\paragraph{Combinatorial Surfaces.~}
A \EMPH{combinatorial surface} $M$ is a surface $\Sigma$ together with an \emph{undirected} multigraph $G$ with non-negative edge weights embedded cellularly on $\Sigma$.  This concept, which we will use in Section~\ref{S:cycles}, was introduced by Colin de Verdi\`ere \cite{c-rcds-03} and it is equivalent to cellular embeddings of graphs.  The combinatorial surface model is dual to the cross-metric surface model; see~\cite{octagons} for more discussion on this.  The \EMPH{complexity} of a combinatorial surface is the size of the multigraph that describes it.

In combinatorial surfaces, all curves are restricted to lie on $G$, possibly with repeated edges or vertices.  Thus, a curve is a walk in $G$. We say a curve is \emph{simple} if it can be infinitesimally perturbed off of $G$ to be simple in $\Sigma$; note that this allows a simple curve to follow the same edge multiple times.  Similarly, two curves in a combinatorial surface have \EMPH{$k$ crossings} if they can be infinitesimally perturbed to generic position with $k$ crossing, and any infinitesimal perturbation of the curves to generic position has at least $k$ crossings.  In all cases, an infinitesimal perturbation of an arc should keep being an arc; the endpoints cannot go away from the boundary of the surface. The \EMPH{multiplicity} of a curve is the maximum number of times that the same edge appears in the curve.  The \EMPH{length} of a curve $\gamma$, denoted by \EMPH{$|\gamma|$}, is the sum of the edge weights of its edges, where each edge weight is counted as many times as the corresponding edge appears in the curve.

\paragraph{Duality.~}
The \EMPH{dual} of a surface-embedded graph $G$ is another graph \EMPH{$G^*$} embedded on the same surface, whose vertices correspond to faces of~$G$ and vice versa.  Two vertices in $G^*$ are joined by an edge if and only if the corresponding faces of $G$ are separated by an edge of $G$; thus, every edge~$e$ in~$G$ has a corresponding dual edge \EMPH{$e^*$} in~$G^*$.  Similarly, every vertex $v$ of $G$ corresponds to a face \EMPH{$v^*$} of $G^*$, and every face $f$ of $G$ corresponds to a vertex \EMPH{$f^*$} of $G^*$.  Each dual edge~$e^*$ is embedded so that it crosses the corresponding primal edge $e$ exactly once and intersects no other primal edge.  Duality is an involution---the dual of $G^*$ is isomorphic to the original graph~$G$.  We will use duality only in orientable surfaces. In this case we will use a duality for darts defined by $\dvec{e} := \dart{(\rsh(\vec{e}))^*}{(\lsh(\vec{e}))^*}$.  (For darts this is not an involution because $((\dart{u}{v})^*)^*=\dart{v}{u}$.)

For any subgraph~$H$ of $G$, let~\EMPH{$H^*$} denote the corresponding subgraph of $G^*$. 

To simplify notation, we will consistently use the letters $s,u,v,w,x,y,z$ to denote vertices in the primal graph $G$, and the letters $a,b,c,d$ to denote vertices in the dual graph $G^*$.  Thus, $a^*$ will always denote a face of $G$.  (However, $e$ is always an edge of $G$, $f$ is always a face of $G$, and $g$ is always the genus of the underlying surface.)  

\paragraph{Tree and Cotrees.~}
Let $G$ be an undirected graph embedded on a surface of genus $g$ without boundary.  A~\emph{spanning tree} of $G$ is just a tree formed by some subset of the edges of $G$ that includes every vertex of~$G$.  If $C^*$ is a tree contained in the dual graph $G^*$, we call the corresponding primal subgraph $C$ a \EMPH{cotree} of~$G$; if moreover $C^*$ is a spanning tree of $G^*$, we call $C$ a \emph{spanning cotree}.  A \EMPH{tree-cotree decomposition} of~$G$ is a partition of $G$ into three edge-disjoint subgraphs: a spanning tree~$T$, a spanning cotree~$C$, and the leftover edges $L = G\setminus (T\cup C)$ \cite{e-dgteg-03}.  For \emph{any} spanning tree $T$, we can complete a tree-cotree decomposition by choosing any spanning tree $C^*$ of the dual subgraph $(G\setminus T)^*$.  Euler's formula implies that $\abs{L} = 2g$; in particular, when the underlying surface is the sphere, $L$ is empty.  

\subsection{Shortest-Path Trees}

\paragraph{Definitions.~}
Let $G = (V,\vec{E})$ be a directed graph, and let $w\colon \vec{E}\to \Real_{\ge 0}$ be a non-negative weight function on the edges.  To shorten some expressions we  define for every edge $uv$ the weight
\[
	\EMPH{$\hw(uv)$} := w(\dart{u}{v}) + w(\dart{v}{u}).
\]
A \EMPH{shortest-path tree} is a spanning tree of $G$ that contains shortest paths from a \EMPH{source} vertex $s$ to every other vertex of $G$.  Shortest-path trees are typically represented by storing two values at every vertex: \EMPH{$\dist(v)$} is the shortest-path distance from $s$ to $v$, and \EMPH{$\pred(v)$} is the predecessor of $v$ in the shortest path from $s$ to $v$, or equivalently, the parent of $v$ if we regard the tree as rooted at~$s$.  In particular, $\dist(s) = 0$ and $\pred(s) = \textsc{Null}$.  

Consider an arbitrary function $\dist\colon V \to \Real_{\ge 0}$, not necessarily the distance from $s$.
We define the \EMPH{slack} of any dart $\dart{u}{v}$ as follows:
\[
	\EMPH{$\slack(\dart{u}{v})$} := \dist(u) + w(\dart{u}{v}) - \dist(v).
\]
So for any edge $uv$, we have 
\[
	\slack(\dart{v}{u}) + \slack(\dart{v}{u}) = w(\dart{u}{v}) + w(\dart{v}{u}) = \hw(uv).
\]
A dart is \EMPH{tense} if its slack is negative, and an undirected edge is tense if either of its darts is tense.  A classical result of Ford \cite{f-nft-56}, exploited by almost all shortest-path algorithms, states that a collection of distances and predecessors describes a shortest-path tree if and only if there are no tense edges, $\dist(s) = 0$, and $\slack(\dart{\pred(v)}{v}) = 0$ for every vertex $v\ne s$. 

It is helpful to view the graph $G$ as a continuous space, so that distances and shortest paths between points in the interior of edges are also well-defined.  Specifically, each edge $uv$ is homeomorphic to an interval~$[0,1]$; each value $\lambda\in[0,1]$ represents a point $s_\lambda$ on $uv$ such that the distance from $s_\lambda$ to $u$ is $\lambda w(\dart{v}{u})$, the distance from $u$ to $s_\lambda$ is $\lambda w(\dart{u}{v})$, the distance from $s_\lambda$ to $v$ is $(1-\lambda) w(\dart{u}{v})$, and the distance from $v$ to $s_\lambda$ is $(1-\lambda) w(\dart{v}{u})$.  In particular, $s_0=u$ and $s_1=v$.  

For simplicity of analysis, we initially assume that all vertex-to-vertex shortest paths are unique, and thus that the union of all shortest paths from a common source vertex define a unique shortest-path tree.  This assumption can be enforced with high probability by adding random infinitesimal weights to each edge~\cite{mvv-memi-87}; we explain this randomized perturbation in detail in Section \ref{S:perturb}, where we also describe a slightly slower deterministic perturbation scheme.

\paragraph{Parametric Shortest Paths.~}

Our algorithm can be viewed as a special case of the \emph{parametric shortest-path} problem introduced by Karp and Orlin \cite{ko-pspaa-81}; see also \cite{ew-oesm-09, parshort, gt-fmccc-89, yto-fpspm-91}.  Suppose we are given a graph whose dart weights are linear functions of a parameter $\lambda$; specifically, let $w_\lambda(\vec{e}) = w(\vec{e}) + \lambda \cdot w'(\vec{e})$, for given functions $w, w'\colon \vec{E}\to \Real_{\ge 0}$.  (In Karp and Orlin's original formulation, $w'(\vec{e})\in\set{0,1}$ for every dart $\vec{e}$.)  Let~$T_\lambda$ denote the shortest-path tree rooted at some fixed source vertex $s$ with respect to the weights~$w_\lambda$.  The parametric shortest-path problem asks us to compute $T_\lambda$ for all $\lambda$ in a certain range.

To solve this problem, Karp and Orlin \cite{ko-pspaa-81} propose maintaining $T_\lambda$ while \emph{continuously} increasing the parameter $\lambda$.  Although the shortest-path distances vary continuously as a function of $\lambda$, the combinatorial structure of $T_\lambda$ changes only at certain \emph{critical values} of $\lambda$.  A critical value occurs when the slack of some non-tree dart $\dart{y}{z}$ decreases to zero; thus, $\dart{y}{z}$ is just about to become tense.  At this critical value, the dart $\dart{y}{z}$ \EMPH{pivots} into the shortest-path tree, replacing some other dart $\dart{x}{z}$; in other words, $y$ becomes the new predecessor of $z$.\footnote{If $z$ is an ancestor of $y$, this pivot introduces a cycle into $T_\lambda$; for all larger values of $\lambda$, the total weight of this cycle is negative, and the shortest-path tree $T_\lambda$ is not well-defined.  In our algorithm, darts will \emph{always} have non-negative weight, so this condition never arises.}

The running time of any parametric shortest-path algorithm clearly depends on two factors: (1) the amortized time required to identify the next tense dart and pivot it into the tree and (2) the number of pivots.  Our initial analysis assumes that exactly one dart becomes tense at each critical value.  Again, this assumption can be enforced with high probability by random infinitesimal perturbation~\cite{mvv-memi-87}, or at the cost of a $O(\log n)$ factor in the worst-case running time; see Section \ref{S:perturb}.

Karp and Orlin's parametric shortest-path algorithm is an early prototypical example of the \emph{kinetic data structure} paradigm of Basch, Guibas, and Hershberger \cite{bgh-dsmd-99, g-kdssar-98}.  Their algorithms (and ours) can also be interpreted as a special case of the parametric objective simplex algorithm of Gass and Saaty \cite{gs-capof-55, kk-ggspc-90}; similar approaches have been applied to several other classical parametric optimization problems \cite{aegh-pkmst-98, es-mters-76, g-pccpp-83}. 

\subsection{Dynamic Forest Data Structures}
\label{SS:dyn-trees}

Our algorithms require dynamic forest data structures that implicitly maintain dart or vertex values under edge insertions, edge deletions, and updates to the values in certain substructures.  We require two different data structures, which we call the \emph{primal tree structure} and the \emph{dual forest structure}, for reasons that will become clear later.

\paragraph{Primal tree structure.~}
Our first data structure maintains a dynamic rooted forest with real values associated with the vertices; in our application, these values are shortest-path distances.  The data structure supports the following operations:
\begin{itemize}
\item
\EMPH{$\create(\val)$:} Return a new tree containing a single vertex with value $\val$.

\item
\EMPH{$\cut(e)$:} Remove the edge $e$ from the tree $T$ that contains it.  The subtree that contains the root of~$T$ keeps that node as its root; the other subtree takes as root the endpoint of $e$ that it contains.

\item
\EMPH{$\join(u,v)$:} Add the edge $uv$ to the forest; the root of the subtree containing $u$ becomes the root of the new larger tree.  This operation assumes that $u$ and $v$ are initially in different trees.

\item
\EMPH{$\getnodevalue(v)$:} Return the value associated with vertex $v$.

\item
\EMPH{$\addsubtree(\Delta, v)$:} Add the real value $\Delta$ to the value of every vertex in the subtree rooted at $v$.

\end{itemize}

We emphasize that the operations $\join(u,v)$ and $\join(v,u)$ give rise to the same undirected tree, but with different roots.  Several dynamic forest data structures, including Euler-tour trees \cite{hk-rfdga-99, tarjan-eulertrees} and self-adjusting top trees~\cite{tw-satt-05}, support each of these operations in $O(\log n)$ amortized time, where $n$ is the number of vertices in the forest.

\paragraph{Dual forest structure.~}
Our second data structure maintains a dynamic undirected, unrooted forest, with two values $\val(\dart{u}{v})$ and $\val(\dart{v}{u})$ associated with every edge $uv$; that is, we associate separate values with each \emph{dart} of the forest.  In our application, we will have an orientable surface, each tree in the dynamic forest is a subgraph of the dual graph~$G^*$, and the value associated with any dart in the forest is equal to the slack of the corresponding primal dart.  The data structure supports the following operations:

\begin{itemize}
\item
\EMPH{$\create(\,)$:} Return a new one-vertex tree.

\item
\EMPH{$\cut(uv)$:} Remove the edge $uv$ from the forest.
 
\item
\EMPH{$\join(\dart{u}{v}, \alpha, \beta)$:} Add the edge $uv$ to the forest and set $\val(\dart{u}{v}) = \alpha$ and $\val(\dart{v}{u}) = \beta$.  This operation assumes that $u$ and $v$ are initially in different trees.  The operations $\join(u, v, \alpha, \beta)$ and $\join(v, u, \beta, \alpha)$ yield identical results.

\item
\EMPH{$\getdartvalue(\dart{u}{v})$:} Returns the value $\val(\dart{u}{v})$.  This operation assumes that $uv$ is an edge in the forest.

\item
\EMPH{$\addpath(\Delta, u, v)$:} For each dart $\dart{x}{y}$ on the (unique) directed path from $u$ to $v$, add~$\Delta$ to $\val(\dart{x}{y})$ and subtract $\Delta$ from $\val(\dart{y}{x})$.  This operation assumes that $u$ and $v$ lie in the same tree. 

\item
\EMPH{$\minpath(u, v)$:} Return a dart $\dart{x}{y}$ on the (unique) directed path from $u$ to $v$, such that $\val(\dart{x}{y})$ is minimized.  This operation assumes that $u$ and $v$ lie in the same tree.

\item
\EMPH{$\junction(t,u,v)$:} Return the unique node that lies on the paths from $t$ to $u$, from $u$ to $v$, and from~$v$ to $t$ in the forest.  This operation assumes that the nodes $t$, $u$, and $v$ lie in the same component of the forest.
\end{itemize}

Several dynamic forest data structures, including link-cut trees \cite{st-dsdt-83}  and self-adjusting top trees~\cite{tw-satt-05}, can be adapted to support each of these operations in $O(\log n)$ amortized time.  In their original formulations, these data structures maintain only a single value at each undirected edge.  However, Goldberg, Grigoriadis, and Tarjan \cite{ggt-udtns-91, t-epnsa-91} describe the necessary modifications to link-cut trees~\cite{st-dsdt-83} to support values on directed edges, specifically for use in network-simplex algorithms, of which our algorithm is an example.  Similar modifications can be applied to any other dynamic tree structure.

The only operation we require that is not part of the standard dynamic-tree literature is \junction.  Here we describe how to implement this operation using link-cut trees; similar methods can be used with more recent dynamic tree data structures.  Internally, link-cut trees actually represent each component of the forest as a \emph{rooted} tree.  Link-cut trees support two additional operations:
\begin{itemize}
\item 
\EMPH{$\evert(v)$:}  Make vertex $v$ the root of the tree that contains it.
\item
\EMPH{$\lca(u,v)$:} Return the least common ancestor of vertices $u$ and $v$.  This operation assumes $u$ and $v$ lie in the same tree \cite[p.~387]{st-dsdt-83}.
\end{itemize}
We can implement $\junction(t,u,v)$ simply by calling $\evert(t)$ and then returning $\lca(u,v)$.

\section{Algorithm for Planar Graphs}
\label{S:planar}

Before describing our algorithm in full generality, we first consider the special case $g=0$.  Given a planar graph $G$ and one of its faces $f$, our algorithm computes an implicit representation of the shortest-path tree rooted at every vertex on the boundary of $f$.  

Our high-level approach is similar to Klein's algorithm \cite{k-msspp-05}.  We begin by computing a shortest-path tree $T$ rooted at an arbitrary vertex on the boundary of $f$.  We maintain $T$ and the complementary dual spanning tree $C^* = (G\setminus T)^*$ in appropriate dynamic forest data structures.  Then, for each edge $uv$ in order around the boundary of $f$, we modify the shortest-path tree rooted at $u$ until it becomes the shortest-path tree rooted at $v$.  However, our algorithm uses a different pivoting strategy to update the tree.  Klein's algorithm moves the source directly from $u$ to $v$ and then performs a sequence of pivots to repair the tree; our algorithm maintains a parametric shortest-path tree as the source point moves \emph{continuously} from $u$ to $v$.

\subsection{Moving Along an Edge}
\label{SS:one-edge}

Consider a single edge $uv$ in $G$.  Suppose we have already computed the shortest-path tree $T_u$ rooted at~$u$.  We transform $T_u$ into the shortest-path tree $T_v$ rooted at $v$ using a parametric shortest-path algorithm as follows.  First, we insert a new vertex $s$ in the interior of $uv$, bisecting it into two edges $su$ and $sv$ with the parametric weights
\begin{alignat*}{2}
	w_\lambda(\dart{u}{s}) &:= \lambda w(\dart{u}{v}),\qquad
	&
	w_\lambda(\dart{s}{v}) &:= (1-\lambda) w(\dart{u}{v}).
	\\
	w_\lambda(\dart{s}{u}) &:= \lambda w(\dart{v}{u}),
	&
	w_\lambda(\dart{v}{s}) &:= (1-\lambda)w(\dart{v}{u}).
\end{alignat*}
Every other dart $\dart{x}{y}$ has constant parametric weight $w_\lambda(\dart{x}{y}) := w(\dart{x}{y})$.  We then maintain the shortest-path tree \EMPH{$T_\lambda$} rooted at $s$, with respect to the weight function $w_\lambda$, as the parameter $\lambda$ increases continuously from $0$ to $1$.  The initial shortest-path tree $T_0$ is equal to $T_u$, and the final tree~$T_{1}$ is equal to~$T_v$.

For any vertex $x$ and any parameter value $\lambda\in [0,1]$, let \EMPH{$\dist_\lambda(x)$} denote the distance from $s$ to $x$ with respect to the weight function $w_\lambda$.  We color each vertex $x$ \textcolor{Red}{\EMPH{red}} if $\dist_\lambda(x)$ is an increasing function of $\lambda$, or \textcolor{Blue}{\EMPH{blue}} if $\dist_\lambda(x)$ is a decreasing function of $\lambda$.\footnote{Readers familiar with astronomy may recall that visible light from objects moving away from the earth is shifted toward red, and visible light from objects moving toward the earth is shifted toward blue.}  The color of a vertex depends on the parameter $\lambda$.  Every vertex except $s$ is either red or blue, for generic values of $\lambda$.  If $su$ is an edge in~$T_\lambda$, then every vertex in the subtree rooted at $u$ is red; symmetrically, if $sv$ is an edge in~$T_\lambda$, then every vertex in the subtree rooted at $v$ is blue.  $T_0$ may contain only red vertices, and $T_1$ may contain only blue vertices, but there is always an interval of values of $\lambda$ such that $T_\lambda$ contains vertices of both colors.  For example it holds $w_\lambda(\dart{s}{u})=w_\lambda(\dart{s}{v})$ when $\lambda = w(\dart{u}{v})/\hw(uv)$.  We color an edge of $G$ \textcolor{Green}{\EMPH{green}} if it has one red endpoint and one blue endpoint.

Similarly, let \EMPH{$\slack_\lambda(\dart{x}{y})$} denote the slack of $\dart{x}{y}$ with respect to the weight function $w_\lambda$:
\[
	\slack_\lambda(\dart{x}{y}) := \dist_\lambda(x) + w_\lambda(\dart{x}{y}) - \dist_\lambda(y).
\]
We call a dart $\dart{x}{y}$ \EMPH{active} if $\slack_\lambda(\dart{x}{y})$ is a decreasing function of $\lambda$; only active darts can become tense.  

The following lemma applies to \emph{arbitrary} graphs, not just planar graphs or even graphs on surfaces.

\begin{lemma}
\label{L:active}
The following claims hold for all real $\lambda \in [0, 1]$.
\begin{itemize}\cramped
\item[(a)] If $su \not\in T_\lambda$, there are no active darts.
\item[(b)] If $sv \not\in T_\lambda$, the only active dart is $\dart{s}{v}$.
\item[(c)] Otherwise, $\dart{x}{y}$ is active if and only if $x$ is blue and $y$ is red.
\end{itemize}
\end{lemma}

\begin{proof}
Consider an arbitrary dart $\dart{x}{y}$.  If $x$ is blue and $y$~is red then, for some constants $\sigma_0,\sigma_1,\sigma_2\ge 0$, 
\begin{align*}
	\slack_\lambda(\dart{x}{y})
	&= \dist_\lambda(x) + w_\lambda(\dart{x}{y}) - \dist_\lambda(y)
	\\
	&= w_\lambda(\dart{s}{v}) + \sigma_1 + w(\dart{x}{y}) - w_\lambda(\dart{s}{u}) - \sigma_2\\
	&= (1-\lambda) w(\dart{u}{v}) + \sigma_1 + w(\dart{x}{y}) -\lambda w(\dart{v}{u}) - \sigma_2\\
	&= \sigma_0 - \lambda \hw(uv),
\end{align*}
and therefore $\dart{x}{y}$ is active.  Symmetrically, if $x$ is red and $y$ is blue, then $\slack_\lambda(\dart{x}{y})$ is an increasing function of $\lambda$, and so $\dart{x}{y}$ is not active.  Finally, if $x$ and $y$ are both red or both blue, then $\slack_\lambda(\dart{x}{y})$ is constant, so $\dart{x}{y}$ is not active.

If $su$ is not an edge in $T_\lambda$, then $sv \in T_\lambda$, and thus every vertex except $s$ is blue.  In this case, $\slack_\lambda(\dart{u}{s})$ is constant and $\slack_\lambda(\dart{s}{u})$ is increasing with $\lambda$, so no dart is active.

Similarly, if $sv$ is not an edge in $T_\lambda$, then $su \in T_\lambda$, and thus every vertex except $s$ is red.  In this case, $\slack_\lambda(\dart{v}{s})$ is constant and $\slack_\lambda(\dart{s}{v})$ is decreasing with $\lambda$, so $\dart{s}{v}$ is the only active dart.
\end{proof}

The proof of Lemma \ref{L:active} has two immediate corollaries, which also apply to arbitrary graphs:

\begin{corollary}
\label{C:next-pivot}
Fix a real value $\lambda$ such that $T_\lambda$ contains both $su$ and $sv$.  As $\lambda$ increases, the next dart to become tense (if any) is the active dart $\dart{x}{y}$ such that $\slack_\lambda(\dart{x}{y})$ is minimized.
\end{corollary}

\begin{proof}
If $\dart{x}{y}$ is active, then $\slack_\lambda(\dart{x}{y}) = \sigma_0 - \lambda \hw(uv)$ for some constant $\sigma_0 \ge 0$.  Thus, the slack of every active dart decreases at exactly the same rate.
\end{proof}

\begin{corollary}
As $\lambda$ increases continuously from $0$ to $1$, the number of pivots in $T_\lambda$ is equal to the number of edges in $T_u\setminus T_v$.
\end{corollary}

\begin{proof}
Just before any dart $\dart{x}{y}$ pivots into the shortest-path tree, $x$ is blue and $y$ is red; after the pivot, both $x$ and $y$ are blue.  Thus, any dart that pivots into $T_\lambda$ never pivots out again.
\end{proof}

\subsection{Pivoting Quickly}

To identify the next pivot quickly when $G$ is planar, we maintain both the shortest-path tree $T_\lambda$ and its complementary spanning cotree $C_\lambda = G \setminus T_\lambda$ in dynamic tree structures.  Specifically, we store $T_\lambda$ in a primal tree structure, where the value of any node $v$ is $\dist_\lambda(v)$, and we store $C_\lambda^*$ in a dual forest structure, where the value associated with any dual dart is the slack of the corresponding primal dart: $\val((\dart{x}{y})^*) := \slack_\lambda(\dart{x}{y})$.

To simplify notation related to the dart $\dart{u}{v}$, let \EMPH{$a$} be the face $\tail((\dart{u}{v})^*) = \rsh(\dart{u}{v})^*$, let \EMPH{$b$} be the face $\head((\dart{u}{v})^*) = \lsh(\dart{u}{v})^*$, and let \EMPH{$\pi_\lambda$} denote the unique directed path from $a$ to~$b$ in the the dual spanning tree $C_\lambda^*$.  See Figure \ref{F:planar}.

\begin{figure}
\centering
\includegraphics[height=1.75in]{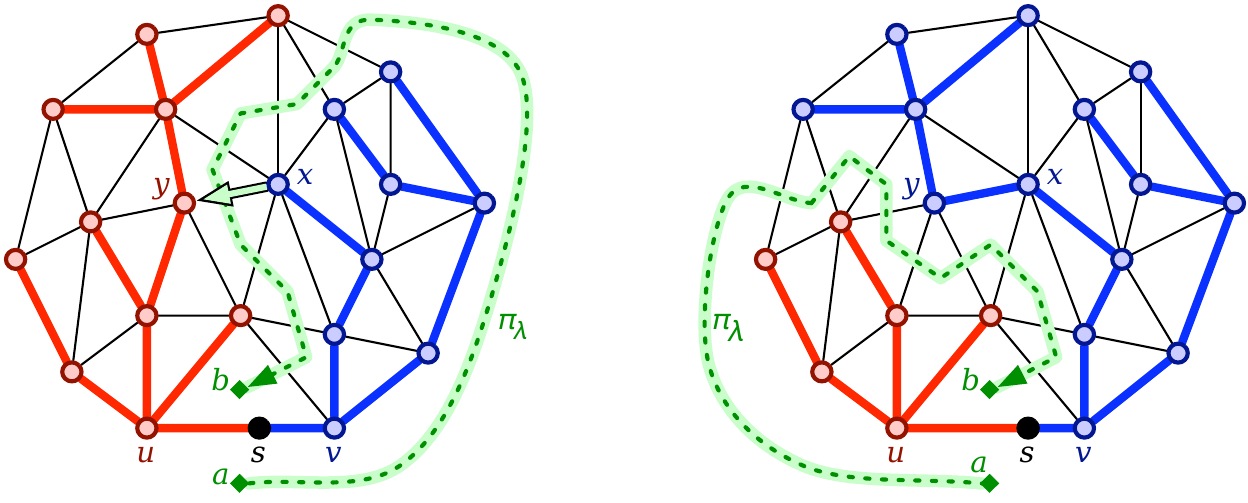}
\caption{A single pivot in a planar shortest-path tree.  Thick (red and blue) lines indicate the shortest-path tree $T_\lambda$; the dotted (green) path is $\pi_\lambda$; the hollow arrow indicates the pivoting dart $\dart{x}{y}$.}
\label{F:planar}
\end{figure}

\begin{lemma}
\label{L:planar-active}
If $T_\lambda$ contains both $su$ and $sv$, then the next dart to become tense (if any) is the dart with minimum slack whose dual lies in the directed path $\pi_\lambda$.
\end{lemma}

\begin{proof}
Lemma \ref{L:active} and Corollary \ref{C:next-pivot} imply that the next dart to become tense is the active dart $\dart{x}{y}$ with minimum slack.  Thus, it suffices to prove that a dart is active if and only if its dual lies in $\pi_\lambda$.

The coloring of the vertices of $G$ extends by duality to a coloring of the faces of $G^*$.  Because the red and blue vertices of $G$ define subtrees of $T_\lambda$, the union of the red dual faces and the union of the blue dual faces are both connected.  Thus, both of these unions are topological disks whose common boundary is a cycle in $G^*$, composed of the edge $ab$ and the unique (undirected) path from~$a$ to~$b$ in~$C^*_\lambda$.  Any active dart must cross this cycle in the opposite direction as the dart $\dart{u}{v}$.
\end{proof}

\begin{lemma}
We can perform the next pivot into $T_\lambda$ in $O(\log n)$ amortized time.
\end{lemma}

\begin{proof}
Our algorithms for finding the next dart to pivot and executing the pivot are shown in Figure \ref{F:planar-pivot}.  Each algorithm performs a constant number of dynamic forest operations, plus a constant amount of additional work.

Under most circumstances, calling $\minpath(a, b)$ gives us (the dual of) the next dart $\dart{x}{y}$ to pivot into the tree, but there are several boundary cases.  Two such cases are described by Lemma \ref{L:active}(a) and~(b); another arises when there is too much slack for $\dart{x}{y}$ to pivot before $\lambda$ reaches $1$.  To detect this latter case, we compute the value $\lambda'$ that would make the slack of $\dart{x}{y}$ zero and check whether it is below $1$.  All these cases are handled by our algorithm \textsc{FindNextPivot}.

If \textsc{FindNextPivot} returns a dart $\dart{x}{y}$, we perform the necessary data structure updates by calling $\textsc{Pivot}(\dart{x}{y})$.  If $\Delta$ is the current slack of $\dart{x}{y}$, the parameter $\lambda$ must increase by $\Delta/\hw(uv)$ before $\dart{x}{y}$ actually pivots into $T_\lambda$.  We increase the distances of the red vertices by $\Delta_R=\Delta w(\dart{v}{u})/\hw(uv)$ by calling $\addsubtree(\Delta_R, u)$, decrease the distances at the blue vertices by $\Delta_B=\Delta w(\dart{u}{v})/\hw(uv)$ by calling $\addsubtree(-\Delta_B, v)$, and adjust the slacks of all the active darts and their reversals by calling $\addpath(-\Delta, a, b)$.  Finally, we fix the underlying tree-cotree decomposition using two \cut\ and two \join\ operations.  
\end{proof}

When there are no edges left to pivot into the tree $T_\lambda$, the algorithm $\textsc{FindNextPivot}$ returns \textsc{Null}.
However, we still have to slide the source $s$ by $1-\lambda$ to reach $v$. Thus, the distances in the primal tree have to be updated by calling $\addsubtree(- (1-\lambda) w(\dart{u}{v})/\hw(uv), v)$ and, if $\pred(u) \ne s$, also $\addsubtree((1-\lambda) w(\dart{v}{u})/\hw(uv), u)$. 
We have shown the following result, already proved by Klein~\cite{k-msspp-05}.

\begin{figure}
\centering\small
\begin{algorithm}
	\textul{\textsc{FindNextPivot}:}\+
\\	if $\pred(v) \ne s$ then return $\dart{s}{v}$
\\	if $\pred(u) \ne s$ then return \textsc{Null}
\\[1ex]
	$(\dart{x}{y})^* \gets \minpath(a, b)$
\\	$\Delta \gets \getdartvalue((\dart{x}{y})^*)$
\\[1ex]
    $\lambda' \gets \lambda+ \Delta/\hw(uv)$
\\[1ex]
	if $\lambda' <1$\+
\\		return $\dart{x}{y}$\-
\\	else\+
\\		return \textsc{Null}\-
\end{algorithm}
\qquad
\begin{algorithm}
	\textul{$\textsc{Pivot}(\dart{x}{y})$:}\+
\\	$\Delta \gets \getdartvalue((\dart{x}{y})^*)$
\\	$\lambda' \gets \lambda+ \Delta/\hw(uv)$ 
\\	$\Delta_B\gets  \Delta w(\dart{u}{v})/ \hw(uv)$
\\  $\Delta_R \gets \Delta w(\dart{v}{u})/\hw(uv)$
\\[1ex]
    if $\pred(u) = s$ then $\addsubtree(\Delta_R, u)$
\\	if $\pred(v) = s$ then $\addsubtree(-\Delta_B, v)$
\\	$\addpath(-\Delta, a, b)$
\\[1ex]
	$z \gets \pred(y)$; ~$\pred(y) \gets x$
\\  $\cut(zy)$;~ $\join(x, y)$
\\	$\cut((xy)^*)$;~ $\join((\dart{z}{y})^*, 0, \hw(yz))$
\end{algorithm}
\caption{Algorithms to find and execute the next pivot when $G$ is planar.}
\label{F:planar-pivot}
\end{figure}

\begin{theorem}
Let $G$ be a directed plane graph with $n$ vertices, and let $s$ be any vertex in $G$.  After $O(n\log n)$ preprocessing time, we can maintain a representation of the shortest-path tree $T_s$ that supports the following operations:
\begin{itemize}\cramped
\item
Given any vertex $v$, return the shortest-path distance from $s$ to $v$ in $O(\log n)$ time.
\item
Given any vertex $v$, return the last edge on the shortest-path from $s$ to $v$ in $O(1)$ time.
\item
For any edge $sv$, change the source vertex from $s$ to $v$ in $O(k\log n)$ amortized time, where $k$ is the number of edges in $T_s\setminus T_v$.
\end{itemize}
\end{theorem}

A nearly identical algorithm updates the shortest-path tree after changing the weight of any single edge, or deleting an edge, or inserting an edge (provided the graph remains planar), in $O(k\log n)$ amortized time, where $k$ is the number of edges added to or deleted from the shortest-path tree.

\subsection{Moving Around a Face}

Finally, we bound the running time of our algorithm as the source vertex moves all the way around the boundary of a given face $f$.  Klein~\cite{k-msspp-05} noted that if one maintains the so-called \emph{leftmost shortest-path tree}, each edge enters and leaves the shortest-path tree at most a constant number of times.  (For graphs with unique shortest-paths, the leftmost shortest-path tree is the unique shortest-path tree.)  This property also holds for our algorithm when the edge weights
are generic.

\begin{lemma}\label{le:edge}
Assume that the edge weights are generic. As the source point $s$ moves around the boundary of $f$, each dart of $G$ pivots into (or out of) the shortest-path tree rooted at $s$ at most once.
\end{lemma}

\begin{proof}
Without loss of generality, assume $f$ is the outer face of $G$.  Let $T$ be any spanning tree $T$ of $G$, let $e$ be any edge of $T$, and let $S$ be either component of $T\setminus e$.  The Jordan curve theorem implies that the subtree~$S$ contains either no vertices of $f$, every vertex of $f$, or a contiguous interval of vertices of $f$.

Consider an arbitrary dart $\dart{x}{y}$.  Let $A$ be the set of vertices of $f$ whose shortest-path trees contain $\dart{x}{y}$.  Equivalently, $A$ is the set of vertices of $f$ that lie in the subtree rooted at $x$ in the shortest-path tree rooted at $y$, if all edges of $G$ are reversed.  Thus, either $A$ is empty, $A$ is the complete set of vertices in $f$, or $A$ is a contiguous interval of vertices along the facial walk defining $f$.  In the first case,  $\dart{x}{y}$ is never in the shortest-path tree; in the second case, $\dart{x}{y}$ is always in the shortest-path tree; in the third case, $\dart{x}{y}$ pivots in exactly once and pivots out exactly once.
\end{proof}

Like Klein \cite{k-msspp-05}, we can assume that the input graph $G$ has bounded degree by replacing every vertex of degree $d$ with a small tree of degree-3 vertices and diameter $O(\log d)$  This assumption allows us to maintain the primal dynamic tree structure and the  predecessor pointers at every node in a persistent data structure~\cite{dsst-mdsp-89}, so that we can access any version of the shortest-path tree in $O(\log n)$ time.  (It is not necessary to make the dual forest structure persistent.)  The resulting persistent data structure requires $O(\log n)$ space per pivot, or $O(n \log n)$ space overall.  When converting the input graph to a graph of bounded degree, any shortest-path with $k$ edges becomes a shortest-path with $O(k\log \Delta)$ edges, where $\Delta$ is the maximum degree in $G$.  We have shown the following result, already proved by Klein~\cite{k-msspp-05} without the assumption on generic weights.

\begin{theorem}
\label{Th:planar}
Let $G$ be a directed plane graph with $n$ vertices and generic edge weights, and let $f$ be any face of $G$.  In $O(n\log n)$ time and space, we can construct a data structure that supports the following operations:
\begin{itemize}\cramped
\item
Given any vertex $u$ on the boundary of $f$ and any other vertex $v$, return the shortest-path distance from $u$ to $v$ in $O(\log n)$ time.
\item
Given any vertex $u$ on the boundary of $f$ and any other vertex $v$, return the shortest-path from $u$ to $v$ in $O(\log n + k \log \Delta)$ time, where $k$ is the number of edges in the path and $\Delta$ is the maximum degree of $G$. 
\end{itemize}
\end{theorem}

\section{Algorithm for Higher-Genus Graphs}
\label{S:surfaces}

Our algorithm for planar graphs does not immediately extend to graphs of higher genus, primarily because the complement of a spanning tree is no longer a cotree, and therefore the active darts can have much more complex structure than a single dual path.  Even when the red and blue subtrees are separated by a single dual cycle at one value of $\lambda$, a single pivot can split the red-blue boundary into multiple dual cycles; see Figure \ref{F:torus}.  Also, the analysis of the planar algorithm (implicitly) relies on the Jordan Curve Theorem, which does not hold on surfaces of higher genus.

We first restrict our attention to \emph{orientable} surfaces. In Section~\ref{S:nonorientable} we show how to reduce the case of non-orientable surfaces to orientable ones.

\begin{figure}
\centering
\includegraphics[height=1.25in]{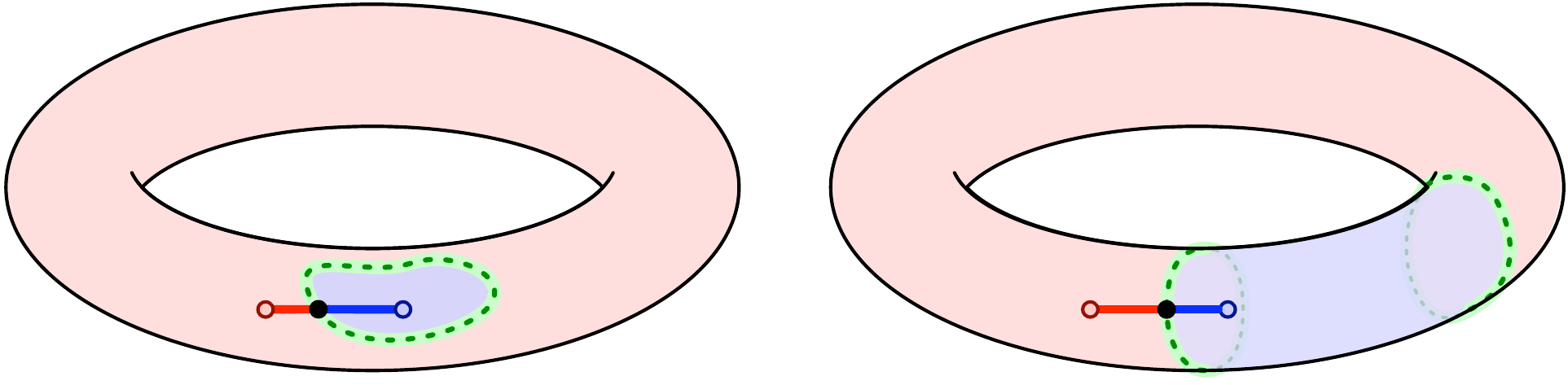}
\caption{On a higher-genus surface, the boundary between the red and blue dual faces may consist of multiple cycles.}
\label{F:torus}
\end{figure}

\subsection{Maintaining a Grove}
\label{SS:grove}

Let $G$ be a cellularly embedded graph on an orientable surface $\Sigma$ of genus $g>0$.  Without loss of generality, we assume that $G$ is a triangulation, so that every vertex of the dual graph $G^*$ has degree $3$.  Let $F$ be a spanning forest of $G$ with $\kappa$ components (that is, a spanning tree of $G$ minus $\kappa-1$ edges) for some constant~$\kappa$.  (In our application, we always have $1\le \kappa \le 3$.)  Euler's formula implies that the complementary dual subgraph $X = (G\setminus F)^*$ is a tree plus $2g + \kappa - 1$ extra edges.  Following Erickson and Har-Peled \cite{schema}, we refer to $X$ as a \EMPH{cut graph}; removing $X$ cuts the underlying surface $\Sigma$ into $\kappa$ topological disks.

Our parametric shortest-path algorithm requires us to quickly identify active darts in this cut graph, and to maintain it as edges are inserted and deleted from $F$.  To support these operations quickly, we decompose $X$ into $O(g)$ edge-disjoint subtrees as follows.

Let $\overline{X}$ be the subgraph of $X$ obtained by repeatedly removing vertices of degree~$1$ until no more remain; $\overline{X}$ is sometimes called the 2-core of $X$.  The subgraph of removed edges is a forest, which we denote $H$ (for `hair').  Equivalently, a dual edge~$e^*$ is in $\overline{X}$ if and only if the endpoints of the primal edge $e$ lie in different trees in $F$, or if both endpoints are in the same tree, but adding $e$ to that tree creates a non-contractible cycle. 

Euler's formula implies that $\overline{X}$ consists of $6g+3\kappa-3 = O(g)$ paths $\pi_1, \pi_2, \pi_3, \dots$, which meet at $4g+2\kappa-2 = O(g)$ vertices of degree~$3$ \cite[Lemma 4.2]{schema}.  We refer to each path~$\pi_i$ as a \EMPH{cut path}, and each degree-3 vertex a \EMPH{branch point}.  For each index $i$, let \emph{$C_i$} denote the union of $\pi_i$ with all trees in $H$ that share a vertex with $\pi_i$.  We call each cut path $\pi_i$ the \EMPH{anchor path} and its endpoints $a_i$ and~$b_i$ the \EMPH{anchor vertices} of the corresponding subtree~$C_i$.  We call the set of $O(g)$ edge-disjoint subtrees $\set{C_1, C_2, C_3, \dots}$ a \EMPH{grove}.\footnote{The Oxford English Dictionary defines \emph{grove} as ``A small wood; a group of trees affording shade or forming avenues or walks, occurring naturally or planted for a special purpose.''}

We maintain the grove by storing the subtrees $C_i$ in the dual forest data structure described in Section \ref{SS:dyn-trees}.  We maintain three separate copies of each branch point in this data structure, one for each subtree~$C_i$ that contains it.  We also separately record (the correct copies of) the anchor vertices of each subtree~$C_i$.  Our grove data structure supports two restructuring operations:
\begin{itemize}
\item
\EMPH{$\GroveJoin(\dart{u}{v}, \alpha, \beta)$:} Add edge $uv$ to the cut graph, and set $\val(\dart{u}{v}) = \alpha$ and ${\val(\dart{v}{u}) = \beta}$.  This operation assumes that $u$ and $v$ are not in the same subtree $C_i$.

\item
\EMPH{$\GroveCut(uv)$:} Remove edge $uv$ from the cut graph.
\end{itemize}

To insert an edge $uv$ into $X$, we first determine the subtrees $C_i$ and $C_j$ that contain $u$ and $v$, respectively.  We then find the vertices $\hat{a} = \junction(u, a_i, b_i)$ and $\hat{b} = \junction(v, a_j, b_j)$.  Next we insert the edge into the dynamic forest by calling $\join(\dart{u}{v}, \alpha, \beta)$, which merges the two subtrees $C_i$ and~$C_j$ into a single subtree $\hat{C}$.  We split $\hat{C}$ into five smaller subtrees by splitting $\hat{a}$ and $\hat{b}$ into three copies, each carrying one outgoing edge, using a constant number of $\cut$ and $\join$ operations.  Finally, we store the anchor vertices of each of the five new subtrees; in particular $\hat{a}$ and $\hat{b}$ are the anchor vertices for the subtree containing $uv$.  See Figure~\ref{F:grove}.  The entire procedure takes $O(\log n)$ amortized time.

To remove the edge $uv$, we follow the insertion procedure in reverse.  We begin by determining which tree $C_i$ contains $uv$.  We then merge the three copies of the anchors $a_i$ and $b_i$ into single vertices, thereby merging five trees into a single tree $\hat{C}$.  We remove the edge from the dynamic forest by calling $\cut(uv)$, splitting $\hat{C}$ into two subtrees, one containing $u$ and the other containing $v$.  Finally, we update the anchor vertices of these two subtrees.  Again, the entire procedure takes $O(\log n)$ amortized time.

\begin{figure}
\centering
\includegraphics[width=\hsize]{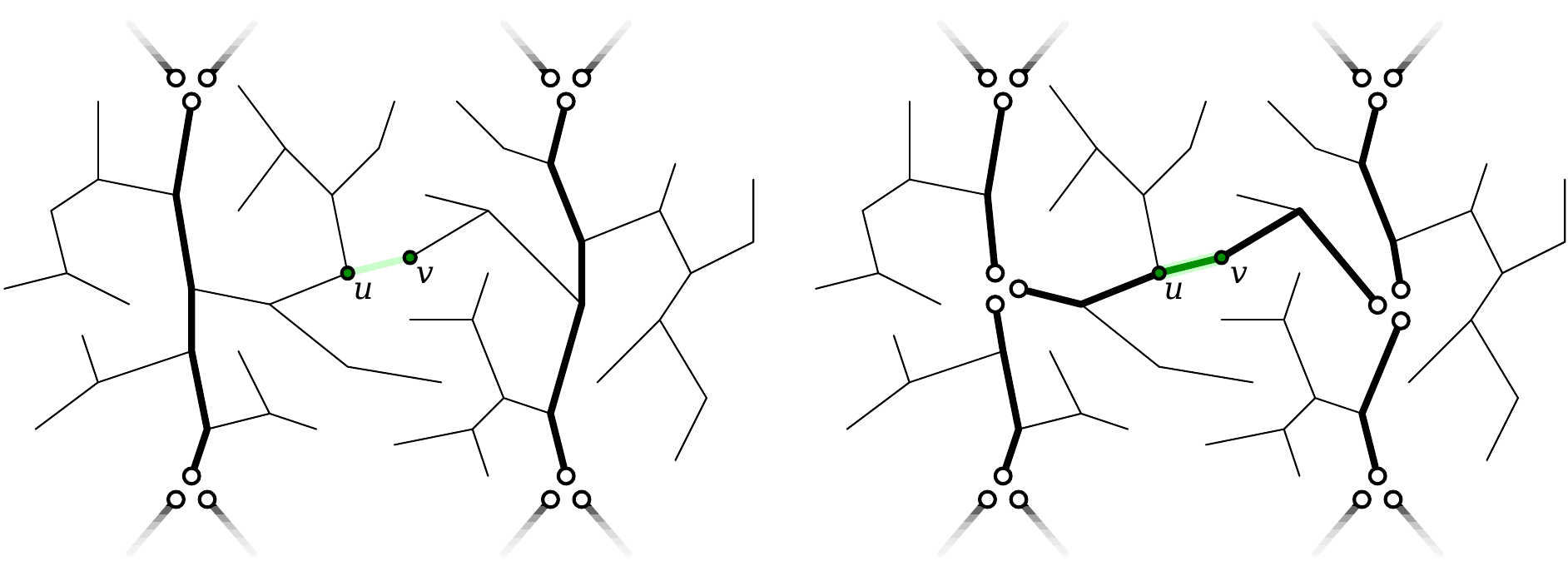}
\caption{Maintaining a grove.  Left to right: Inserting an edge.  Right to left: Deleting an edge.}
\label{F:grove}
\end{figure}

\subsection{Pivoting Quickly}
\label{SS:fast-pivot}

Now we describe how to efficiently move the source of a shortest-path tree in $G$ along a single edge~$uv$, using the parametric shortest-path approach described in Section \ref{SS:one-edge}.  We cannot apply the planar algorithm directly, because the set of active darts is no longer dual to a directed path in $G^*$.  However, the grove structure we just described lets us quickly identify a superset of the active darts.

Suppose $su$ and $sv$ are both edges in the current shortest-path tree $T_\lambda$.  Let $R$ and $B$ be the induced red and blue subtrees of $T_\lambda$, and let $X = (G \setminus (R \cup B))^*$.  We maintain the cut graph $X$ in a grove, as described above.  For each $i$, all the primal edges dual to the edges on the anchor path $\pi_i$ have the same color: the paths $\pi_1,\pi_2,\dots$ cut the surface into two topological disks, one of them containing the blue vertices and the other containing the red vertices, and thus each path $\pi_i$ is on the boundary between disks of different colors or disks of the same color.  We can then associate a color to each subtree $C_i$ in the grove according to the colors of primal edges that cross the anchor path $\pi_i$.  If the crossing edges have two red endpoints, we color $C_i$ \emph{red}; if the crossing edges have two blue endpoints, we color $C_i$ \emph{blue}; if the crossing edges have one red endpoint and one blue endpoint, we color $C_i$ \emph{green}.  Observe that the green edges of $G$ are precisely the edges that cross green anchor paths.

We direct the green anchor paths so that they contain the duals of active darts; that is, for any dart $\dart{x}{y}$ where $x$ is blue and $y$ is red, the dual dart $(\dart{x}{y})^*$ lies on some directed green anchor path.  For each directed green anchor path $\pi_i$, let $a_i$ denote its start vertex and $b_i$ its end vertex.  The following lemma is now almost immediate.

\begin{lemma}
We can perform the next pivot into $T_\lambda$ in $O(g \log n)$ amortized time.
\end{lemma}

\begin{proof}
The modified algorithms to find the next edge to become tense and pivot it into the shortest-path tree are shown in Figure \ref{F:g-pivot}.  There are only a few differences from the planar algorithms.  First, we must examine every directed green anchor path to find the next edge to pivot.  Second, we must call \addpath\ on every green anchor path to update the slacks of the green edges.  Finally, we must call \GroveCut\ and \GroveJoin\ to repair the grove data structure.  Since there are $O(g)$ subtrees $C_i$, these algorithms perform $O(g)$ dynamic forest operations, plus constant additional work, so their amortized running time is $O(g\log n)$.
\end{proof}

\begin{figure}
\centering\small
\begin{algorithm}
	\textul{\textsc{FindNextPivot}:}\+
\\	if $\pred(v) \ne s$ then return $\dart{s}{v}$
\\	if $\pred(u) \ne s$ then return \textsc{Null}
\\[1ex]
	$\emph{minStep} \gets 1-\lambda$
\\	$\emph{nextPivot} \gets \textsc{Null}$
\\[1ex]
	for each subtree $C_i$\+
\\		if $C_i$ is green\+
\\			$(\dart{x_i}{y_i})^* \gets \minpath(a_i, b_i)$
\\			$\Delta \gets \getdartvalue((\dart{x_i}{y_i})^*)$
\\			if $\Delta /\hw(uv) < \emph{minStep}$\+
\\				$\emph{nextPivot} \gets (\dart{x_i}{y_i})$
\\				$\emph{minStep} \gets \Delta /\hw(uv)$\-\-\-
\\[1ex]
	return \emph{nextPivot}\-
\end{algorithm}
\qquad
\begin{algorithm}
	\textul{$\textsc{Pivot}(\dart{x}{y})$:}\+
\\	$\Delta \gets \getdartvalue((\dart{x}{y})^*)$
\\	$\lambda \gets \lambda + \Delta/\hw(uv)$ 
\\	$\Delta_B\gets  \Delta w(\dart{u}{v})/ \hw(uv)$
\\  $\Delta_R \gets \Delta w(\dart{v}{u})/\hw(uv)$
\\[1ex]
	if $\pred(u) = s$ then $\addsubtree(\Delta_R, u)$
\\	if $\pred(v) = s$ then $\addsubtree(-\Delta_B, v)$
\\[1ex]
	for each subtree $C_i$\+
\\		if $C_i$ is green\+
\\			$\addpath(-\Delta, a_i, b_i)$\-\-
\\[1ex]
	$z \gets \pred(y)$; $\pred(y) \gets x$
\\  $\cut(zy)$;\quad $\GroveJoin((\dart{z}{y})^*, 0, \hw(yz))$
\\  $\join(x, y)$;~ $\GroveCut((xy)^*)$
\end{algorithm}
\caption{Naive algorithms to find and execute the next pivot when $G$ has positive genus.}
\label{F:g-pivot}
\end{figure}

With a little more effort, we can improve the running time slightly.  In addition to our other data structures, we maintain a priority queue of the green subtrees in the grove, where the priority of any subtree $C_i$ is the minimum slack among all darts whose duals lie in $\pi_i$.  Because the grove never contains more than $O(g)$ subtrees, each priority queue operation takes $O(\log g)$ time.

To find the next dart to pivot, instead of looping over every subtree in the grove, we simply extract the minimum element from the priority queue.  Within each call to \GroveCut\ or \GroveJoin, we perform a constant number of priority queue operations as subtrees are created and destroyed.

Finally, instead of calling $\addpath$ once for each green subtree in \textsc{Pivot}, we maintain a global offset for each green subtree in its priority queue record.  This allows us to call \addpath\ just once on each green subtree, just before it is removed from the priority queue---when the subtree changes color, when the subtree is destroyed by \GroveCut\ or \GroveJoin, or when the moving source point $s$ reaches $v$.  Thus, crudely, the number of calls to \addpath\ is less than the number of priority queue operations.

If we ignore priority queue operations and calls to \addpath, the remainder of \textsc{FindNextPivot} and \textsc{Pivot} requires $O(\log n)$ amortized time.  Thus, to complete our analysis, we need only an upper bound on the number of priority queue operations.

A single pivot can change the colors of several subtrees in the grove.  When a subtree changes from red to green, we must insert it into the priority queue; when it changes from red to blue or from green to blue, we must delete it.  To detect these color changes quickly, we separately maintain a \EMPH{reduced cut graph $\widetilde{X}$}, an abstract graph with a vertex for every branch point of $X$ and an edge $\widetilde{e}_i$ for every cut path~$\pi_i$, with a cellular embedding on $\Sigma$ consistent with the embedding of $\overline{X}$.  Each face of the reduced cut graph corresponds to a component of the primal forest $F$.  We can easily update the reduced cut graph in $O(1)$ time within each call to \GroveCut\ or \GroveJoin. 

Just after removing the old edge $\dart{z}{y}$ from the shortest-path tree $T_\lambda$ and adding its dual to the grove, color the subtree rooted at $y$ \textcolor{Purple}{\EMPH{purple}}.  Now the reduced cut graph $\widetilde{X}$ has exactly three faces: one red, one blue, and one purple.  Edges $\widetilde{e}_i$ on the boundary of the purple face correspond to subtrees $C_i$ that are changing color; specifically:
\begin{itemize}
\item
If $\widetilde{e}_i$ also bounds the red face, the corresponding subtree $C_i$ is changing from red to green; we insert it into the priority queue.
\item
If $\widetilde{e}_i$ also bounds the blue face, the corresponding subtree $C_i$ is changing from green to blue; we delete it from the priority queue.
\item
If $\widetilde{e}_i$ has the purple face on both sides, the corresponding subtree $C_i$ is changing directly from red to blue; we ignore it.  (However, for purposes of analysis, we pretend to execute two priority queue operations.)
\end{itemize}
During a single pivot, the time to update $\widetilde{X}$ and traverse the purple face is less than the time spent maintaining the priority queue, so we can ignore it.

Now recall that the grove always contains $O(g)$ subtrees.  Let $n_{\text{red}}$ and $n_{\text{blue}}$ denote the number of red and blue subtrees, respectively.  The number of priority queue operations for a single pivot is a constant (less than 20) plus the \emph{decrease} in $n_{\text{red}}$ plus the \emph{increase} in $n_{\text{blue}}$.  Over the entire parametric shortest-path algorithm, the total decrease in $n_{\text{red}}$ and the total increase in $n_{\text{red}}$ is at most $O(g)$.

Thus, if moving the source point across $uv$ requires $k$ pivots, the total number of priority queue operations is $O(g+k)$.  Because we perform at most one call to \addpath\ for each priority queue operation, the number of calls to \addpath\ is also at most $O(g+k)$.  We conclude:

\begin{theorem}
\label{Th:g-edge}
Let $G$ be a directed graph with $n$ vertices, cellularly embedded on an orientable surface of genus~$g$, and let $s$ be any vertex in $G$.  After $O(n\log n)$ preprocessing time, we can maintain a representation of the shortest-path tree $T_s$ that supports the following operations:
\begin{itemize}\cramped
\item
Given any vertex $v$, return the shortest-path distance from $s$ to $v$ in $O(\log n)$ time.
\item
Given any vertex $v$, return the last edge on the shortest path from $s$ to $v$ in $O(1)$ time.
\item
For any edge $sv$, change the source vertex from $s$ to $v$ in $O((g+k)\log n)$ amortized time, where $k$ is the number of edges in $T_s\setminus T_v$.
\end{itemize}
\end{theorem}

Again, nearly identical algorithms update the shortest-path tree after changing the weight of any single edge, or deleting an edge, or inserting an edge (provided the graph remains embedded), in $O((g+k)\log n)$ amortized time, where $k$ is the number of edges added to or deleted from the shortest-path tree.

\subsection{Moving Around a Face}

We next bound the running time of our algorithm as the source vertex moves all the way around the boundary of a given face $f$.  We first show a bound on the number of times that an edge can enter or leave the shortest-path tree during the entire traversal, assuming generic edge weights. 

\begin{lemma}
\label{L:edge-g} Assume that the edge weights are generic.
As the source point $s$ moves around the boundary of $f$, each dart of $G$ pivots into (or out of) the shortest-path tree rooted at $s$ at most $O(g)$ times.
\end{lemma}

\begin{proof}
Fix an arbitrary dart $\dart{x}{y}$.  Consider the facial walk $W_f$ defining $f$ and assume that the source is moved along $W_f$. Thus, some edges may be traversed twice by the source of the shortest-path tree, once in each direction. Let $A$ be the set of points $s$ (not just vertices) along $W_f$ such that the shortest-path tree rooted at $s$ contains the dart $\dart{x}{y}$.  $A$ is the union of a set of disjoint, maximal segments $A_1, \dots, A_r$ on $W_f$.  Dart $\dart{x}{y}$ enters the shortest-path tree exactly $r$ times, once at the initial endpoint of each segment $A_i$.

Fix a point $v_i$ in each segment $A_i$, and let $p_i$ be the shortest path from $v_i$ to $y$.  By construction, $p_i$ uses the dart $\dart{x}{y}$, and for all $i \ne j$, paths $p_i$ and $p_j$ do not cross.

Let $\Sigma/f$ be the surface obtained by contracting the face $f$ to a point.  Suppose $p_i$ and $p_j$ are homotopic in $\Sigma/f$.  Then in $\Sigma$, there is a subwalk $f'$ of $f$ that together with $p_i$ and $p_j$ bound a disk (and thus a planar graph).  This implies that the shortest path to $y$ from every vertex of $f'$ must contain the dart $\dart{x}{y}$, because of Lemma~\ref{le:edge}.  It follows that $v_i$ and $v_j$ lie in the same segment of $A$, which is only possible if $i$ and $j$ must be equal.  We conclude that the paths $p_1, p_2, \dots, p_r$ are pairwise non-homotopic.

Finally, a double-counting argument using Euler's formula implies that the maximum possible number of pairwise non-crossing, non-homotopic paths in $\Sigma/f$ is $O(g)$ \cite[Lemma 2.1]{splitting}.  Thus, $r = O(g)$, and the proof is complete.
\end{proof}

Recall from Theorem \ref{Th:g-edge} that the time to move the shortest-path tree across a single edge is $O((g+k)\log n)$, where $k$ is the number of pivots.  The previous lemma implies that the total number of pivots is $O(gn)$, and the total number of edges in $f$ is trivially $O(n)$.  Thus, the total running time to move the source of the shortest-path tree along the boundary of a face is $O(g n \log n)$.

Finally, in order to make our primal data structure persistent, we require the graph to have bounded degree; however, our grove data structure requires the graph to be a triangulation.  We can escape this seeming contradiction as follows.  Given an \emph{arbitrary} embedded graph $G$, we first replace every high-degree vertex with a small tree of degree-3 vertices, each consisting of darts with infinitesimal weight, and then we triangulate each face with edges whose darts have large enough weight.  The edges whose darts have large weight never participate in any shortest path, so the maximum degree of any node in the shortest-path path tree is always at most $3$, and so the primal tree structure can be made persistent efficiently.  We conclude:

\begin{theorem}
\label{Th:genus}
Let $G$ be an directed graph with $n$ vertices and generic edge weights, cellularly embedded in an orientable surface of genus~$g$, and let $f$ be any face of $G$.  In $O(g n\log n)$ time and space, we can construct a data structure that supports the following operations:
\begin{itemize}\cramped
\item
Given any vertex $u$ on the boundary of $f$ and any other vertex $v$, return the shortest-path distance from $u$ to $v$ in $O(\log n)$ time.
\item
Given any vertex $u$ on the boundary of $f$ and any other vertex $v$, return the shortest path from $u$ to $v$ in $O(\log n + k \log \Delta)$ time, where $k$ is the number of edges in the path and $\Delta$ is the maximum degree of $G$. 
\end{itemize}
\end{theorem}

\subsection{Non-Orientable Surfaces}
\label{S:nonorientable}

We can extend Theorem~\ref{Th:genus} to non-orientable surfaces working in the \emph{orientable double cover} of the surface, which is an orientable covering space of $\Sigma$.  (Our construction is essentially the same as used in~\cite{cdem-fotc-10,cm-fsnnc-07}, but is presented here for completeness.)

Recall that a cellular graph embedding of a graph $G$ on a non-orientable surface consists of a rotation system, encoding the cyclic sequence of outgoing darts from each vertex, and a signature, consisting of a bit $i(e)$ for each edge $e$. The signature indicates whether the cyclic sequences at its endpoints are in the same direction or the opposite direction.  

The \EMPH{orientable double cover} $D$ of $G$ is an embedded graph constructed as follows.  For each vertex $v$ in the original graph, create two vertices $(v,0)$ and $(v,1)$ in $D$.  For each dart $\dart{u}{v}$ in the original graph, create two darts $(\dart{u}{v},0) = \dart{(u,0)}{(v,i(uv))}$ and $(\dart{u}{v},1) = \dart{(u,1)}{(v,1-i(uv))}$ in $D$.  The rotation system at $(v,0)$ and $(v,1)$ is just a copy of the rotation system of $v$.  That is, if $\pi_G$ is the original rotation system and $\pi_D$ the rotation system in $D$, then for any dart $\dart{u}{v}$ of $G$ it holds $\pi_D(\dart{u}{v},j)= (\pi_G(\dart{u}{v}),j)$.  Finally, darts $(\dart{u}{v},0)$ and $(\dart{u}{v},1)$ have the same dart weight and the same signature bit as $\dart{u}{v}$.  This cover is orientable because, for any cycle, the sum of the signatures of its edges is $0$ modulo 2.  Note that this cover $D$ has twice as many edges as $G$ and can be constructed in linear time.  

Any walk $(u_0,j_0)(u_1,j_1)(u_2,j_2)\dots(u_t,j_t)$ in $D$ naturally \EMPH{projects} to the walk $u_0u_1u_2\dots u_t$ in $G$.  A \EMPH{lift} of a walk $u_0u_1u_2\dots u_t$ in $G$ is a walk $(u_0,j_0)(u_1,j_1)(u_2,j_2)\dots(u_t,j_t)$ in $D$.  That is, the projection of a lift recovers the original walk.  A lift is uniquely determined by the value $j_0$.  Indeed, since $(u_\ell,j_\ell)(u_{\ell+1},j_{\ell+1})$ must be an edge of $D$, we have 
$$
	j_{\ell+1}=j_{\ell}+i(u_\ell u_{\ell+1})=j_0 + \sum_{\ell'\le \ell} i(u_{\ell'} u_{\ell'+1}) \qquad \mbox{(modulo $2$)}.
$$
This implies that a shortest path from $u$ to $v$ in $G$ corresponds to either the shortest path from $(u,j)$ to $(v,0)$ or the shortest path from $(u,j)$ to $(v,1)$, where $j\in \set{0,1}$ is arbitrary.  Therefore, the distance in $G$ from $u$ to $v$ is the minimum between the distance from $(u,j)$ to $(v,0)$ and the distance from $(u,j)$ to $(v,1)$. 

It is easy to see that a cycle in $D$ is a facial cycle of $D$ if and only it projects to a facial cycle of $G$.  It follows that $D$ has twice as many faces as $G$. 
Euler's formula then implies that the genus of $D$ is $g-1$.  

We can now extend Theorem~\ref{Th:genus} to non-orientable surfaces. Given the graph $G$ and a face $f$, we construct the orientable double cover $D$, choose one of the lifts $f'$ of $f$ in the double cover, and store the data structure of Theorem~\ref{Th:genus} for $D$ and $f'$.  This construction takes $O(gn\log n)$ time because $D$ has size $O(n)$ and smaller genus.  Whenever we want to compute a distance in $G$ from a vertex $u$, on the boundary of $f$, to a vertex $v$, we take the copy $(u,j)$ of $s$ on the boundary of $f'$, query for the distances in $D$ from $(u,j)$ to $(v,0)$ and from $(u,j)$ to $(v,1)$, and return the smallest.  This takes $O(\log n)$ time.  The shortest path from $s$ to $v$ can also be recovered using the projection of the shortest path in $D$ from $(u,j)$ to $(v,0)$ or $(v,1)$, whichever is closer.  We summarize:

\begin{corollary}
Theorem~\ref{Th:genus} also holds for non-orientable surfaces.
\end{corollary}

\section{Computing Shortest Non-Separating and Non-Contractible Cycles}
\label{S:cycles}

Let $\Sigma$ be a combinatorial surface with complexity $n$ and genus $g$; this surface may or may not be orientable.  In this section, we describe algorithms to find the shortest non-separating and shortest non-contractible cycles in $\Sigma$.  Throughout this section we assume generic edge weights. Arbitrary weights are treated in Section~\ref{S:perturb}.  Our use of the technique for maintaining shortest-path trees is condensed in the following lemma.

\begin{lemma}\label{le:arc}
Let $\alpha$ be an arbitrary simple cycle or arc in $\Sigma$.  A shortest cycle crossing $\alpha$ exactly once can be obtained in $O(g n\log n)$ time.
\end{lemma}

\begin{proof}
Consider the surface obtained by cutting $\Sigma$ along $\alpha$: each vertex $v$ in $\alpha$ gives rise to two vertices~$v'$ and $v''$, and two boundary arcs or cycles $\alpha'$ and $\alpha''$.  Let $\Sigma'$ be the surface obtained by gluing disks to the boundary components that contain $\alpha'$ and $\alpha''$.  (If $\alpha$ is an arc or a 1-sided cycle, then $\alpha'$ and $\alpha''$ can be contained in a single boundary.)  A cycle in $\Sigma$ that crosses $\alpha$ once at a point $v$ becomes a path in~$\Sigma'$ between $v'$ and $v''$.  Thus, a shortest cycle that crosses $\alpha$ once at $v$ is a shortest path that connects~$v'$ to~$v''$ in $\Sigma'$, and vice versa.  All the points $v'$ with $v\in \alpha$ belong to a face of $\Sigma'$.  Thus, Theorem~\ref{Th:genus} implies that  we can find a closest pair $(v_0', v_0'')$ in $O(g n\log n)$ time.  Computing the shortest path from $v'_0$ to $v''_0$ gives the desired path.
\end{proof}

For any simple arc or cycle $\alpha$, let $\mathit{Cross}(\alpha)$ denote 
the set of cycles that cross $\alpha$ exactly once.  If $\alpha$ is separating, then $\mathit{Cross}(\alpha)$ is empty, because every cycle crosses $\alpha$ an even number of times.  Also, every cycle in $\mathit{Cross}(\alpha)$ is non-contractible, because contractible cycles are also separating, and therefore any cycle must cross them an even number of times.

\subsection{Shortest Non-Separating Cycle}

Consider a surface $\Sigma$. It suffices to consider surfaces without boundary, because any shortest non-separating cycle in $\Sigma$ is also non-separating in the surface obtained by attaching disks to the boundaries.

Cabello and Mohar~\cite{cm-fsnnc-07} describe how to construct in $O(gn\log n)$ time a set $S$ of $O(g)$ simple cycles such that the shortest cycle in $\bigcup_{\ell\in S} \mathit{Cross}(\ell)$ is a shortest non-separating cycle.  This set is constructed by fixing a shortest-path tree $T_x$ from any vertex $x \in G$,  fixing a tree-cotree decomposition $(T_x, C_x, L_x)$, and then setting $S $ to be the set of cycles formed by adding each of the edges $e \in L_x$ to $T_x$.  Cabello and Mohar use $\mathbb{Z}_2$ homology and the fact that the cycles are formed from 2 shortest paths to prove that the shortest non-separating cycle must cross some cycle of $S$ exactly once~\cite{cm-fsnnc-07}.

Once we have the set $S$, we compute the shortest non-separating cycle by applying Lemma~\ref{le:arc} once for each simple cycle in $S$, and taking the globally shortest cycle.

\begin{theorem}\label{th:non-sep}
Let $\Sigma$ be a surface of complexity $n$ and genus~$g$, with generic weights.  We can find a shortest non-separating cycle in $\Sigma$ in $O(g^2 n\log n)$ time.
\end{theorem}

\subsection{Shortest Non-Contractible Cycle}

The main technique for non-contractible cycles is to find a curve that intersects a shortest non-contractible cycle at most once and whose removal decreases either the genus or the number of boundaries of $\Sigma$.  Our main tool to prove that such a curve exists is the following exchange argument.

\begin{lemma}\label{le:crossings}
Let $\Sigma$ be a combinatorial surface and let $\ell_x$ be a shortest non-contractible loop with given basepoint $x\in \Sigma$.  There is a shortest non-contractible cycle in $\Sigma$ that crosses $\ell_x$ at most once.
\end{lemma}

\begin{proof}
Let $C$ be any non-contractible cycle that crosses $\ell_x$ at least twice, at points $y$ and $z$.  Let $\gamma_1$ and~$\gamma_2$ be the two subpaths of $C$ from $y$ to $z$, and let $\beta_1$ and $\beta_2$ be the subpaths of $\ell_x$ from $y$ to $z$ so that $\beta_1$ contains $x$. To simplify notation, we do not differentiate between a path and its reverse.  We consider two cases, and in each case find a non-contractible cycle $\tilde{C}$ that is no longer than $C$ and crosses $\ell_x$ fewer times than $C$. This implies that some shortest non-contractible cycle crosses $\ell_x$ at most once.

First, suppose $\beta_2 \sim \gamma_1$; the case $\beta_2 \sim \gamma_2$ is symmetric.  The loop $\beta_1 \concat \gamma_1$ passes through $x$, and it is non-contractible because $(\beta_1 \concat \gamma_1) \sim (\beta_1 \concat \beta_2) = \ell_x$.  We then have $|\beta_2| \le |\gamma_1|$ because $|\ell_x|\leq |\beta_1|+|\gamma_1|$.  The cycle $\tilde C=\gamma_2 \concat \beta_2$ is non-contractible because $(\gamma_2\concat\beta_2) \sim (\gamma_2 \concat \gamma_1) = C$.  The cycle $\tilde{C}$ also crosses $\ell_x$ at least two fewer times than $C$, and $|\tilde C|= |\gamma_2|+|\beta_2| \le |\gamma_2|+|\gamma_1| \leq |C|$.

Now suppose that $\beta_2$ is not homotopic to $\gamma_1$ or $\gamma_2$.  In this case, the cycles $\beta_2\concat\gamma_1$ and $\beta_2\concat\gamma_2$ are non-contractible.  The cycle $\beta_1\concat\gamma_1$ or the cycle $\beta_1\concat\gamma_2$ is non-contractible (perhaps both); otherwise $\gamma_1 \sim \beta_1 \sim \gamma_2$, which would imply that $C=\gamma_1\concat \gamma_2$ is contractible.  Let us suppose that $\beta_1\concat\gamma_1$ is non-contractible; the other case is symmetric.  Since $\beta_1 \concat\gamma_1$ is a non-contractible loop through $x$, we have $|\beta_1|+|\beta_2|= |\ell_x| \leq |\beta_1|+|\gamma_1|$, and so $| \beta_2 | \leq | \gamma_1|$.  The cycle $\tilde C = \beta_2\concat\gamma_2$ is non-contractible because $\beta_2\nsim \gamma_2$.  Cycle $\tilde C$ also crosses $\ell_x$ two fewer times than $C$, and $|\tilde C|=|\beta_2|+|\gamma_2|\leq |\gamma_1|+|\gamma_2| = |C|$.
\end{proof}

The following lemma discusses what happens with simple non-contractible cycles when pasting a disk into a boundary of the surface. It should be noted that any shortest non-contractible cycle is always simple.

\begin{lemma}\label{le:gluedisk}
    Let $\Sigma$ be a surface with boundary and let $\delta$ be one of its boundary components.
    Let $\Sigma'$ be the surface obtained by pasting a disk to $\delta$.
    A non-contractible simple cycle in $\Sigma$ is either non-contractible in $\Sigma'$ or
    homotopic to $\delta$ in $\Sigma$.
\end{lemma}
\begin{proof}

Consider a non-contractible simple cycle $C$ in $\Sigma$. Let $D_\delta$ be the disk that is attached to $\Sigma$ to obtain~$\Sigma'$.  If $C$ is contractible in $\Sigma'$, then $C$ bounds a disk $D_C$ in $\Sigma'$.  The disk $D_C$ must contain $D_\delta$, because otherwise, $C$ would also bound a disk in $\Sigma$, implying that $C$ is contractible in $\Sigma$.  $D_C\setminus D\delta$ is an annulus in $\Sigma$ with boundary cycles $C$ and $\delta$.  It follows that $C$ and $\delta$ are homotopic in $\Sigma$.
\end{proof}

We also have the following results regarding arcs in a surface with
boundary.

\begin{lemma}\label{le:crossings2}
Let $\Sigma$ be a surface with boundary, let $\delta$ be one of its boundary cycles, and let $\alpha$ be a shortest non-contractible arc with endpoints in $\delta$.  There is a shortest non-contractible cycle in $\Sigma$ that is either homotopic to $\delta$ or that crosses $\alpha$ at most once.
\end{lemma}

\begin{proof}
Let $C$ be a shortest non-contractible cycle in $\Sigma$. Assume $C \nsim \delta$, since otherwise the proof is complete.  Lemma~\ref{le:gluedisk} implies that $C$ is a shortest non-contractible cycle in $\Sigma'$, the surface obtained by pasting a disk $D$ at $\delta$. In the surface $\Sigma'$, place a new vertex $p$ in $D$ and connect it through edges with weight $L$ to each vertex of $\delta$.  Consider the loop $\ell$ with basepoint $p$ that follows the edge $p\, \alpha(0)$, the arc~$\alpha$, and the edge $\alpha(1)\, p$.  If we choose $L$ large enough, $\ell$ is a shortest non-contractible loop in $\Sigma'$ through $p$, so Lemma~\ref{le:crossings}(a) implies that a shortest non-contractible cycle $C'$ in $\Sigma'$ crosses $\ell$ at most once. We must have $|C'|=|C|$, and if $L$ is large enough, $C'$ must avoid the disk $D$. It follows that $C'$ is a shortest non-contractible cycle in $\Sigma$ that crosses $\alpha$ at most once.
\end{proof}

\begin{lemma}\label{le:crossings3}
Let $\Sigma$ be a surface with at least two boundary components, and let $\alpha$ be a shortest arc connecting two different boundaries of $\Sigma$. There is a shortest non-contractible cycle in $\Sigma$ that crosses $\alpha$ at most once.
\end{lemma}

\begin{proof}
Let $C$ be a shortest non-contractible cycle in $\Sigma$ that crosses
$\alpha$ the minimum number of times. We claim that $C$ crosses
$\alpha$ at most once. Assume for the purpose of contradiction that
$C$ crosses $\alpha$ at least twice, and let $y$ and $z$ be two such
crossings. Let $\gamma_1$ and $\gamma_2$ be the two subpaths of $C$
from $y$ to $z$, let $\tilde\alpha$ be the subpath of $\alpha$ from
$y$ to $z$.  Since $\alpha$ is a shortest arc connecting the two
specified boundaries, we have $|\tilde\alpha|\leq |\gamma_1|$ and
$|\tilde\alpha|\leq |\gamma_2|$.  If $\tilde\alpha \nsim \gamma_1$,
then we have a contradiction: the cycle $\tilde\alpha \concat \gamma_1$ is
non-contractible, crosses $\alpha$ fewer times than $C$ does, and
$|\tilde\alpha|+|\gamma_1|\leq |\gamma_2|+|\gamma_1|=|C|$.
Similarly, we must have $\tilde\alpha \sim \gamma_2$.  But then
$\gamma_1 \sim \gamma_2$, which implies that $C$ is contractible,
which is a contradiction.
\end{proof}

The next lemma summarizes the algorithmic tools that we will use.

\begin{lemma}\label{le:find}
Let $\Sigma$ be a surface of complexity $n$.
\begin{itemize}\cramped
\item[\textup{(a)}]
Given a basepoint $x$, we can find a shortest non-contractible loop with basepoint~$x$ in $O(n\log n)$ time. This loop has multiplicity 2.

\item[\textup{(b)}]
Given a boundary $\delta$, we can find in $O(n\log n)$ time a shortest cycle homotopic to $\delta$.

\item[\textup{(c)}]
Given a boundary $\delta$, we can find in $O(n\log n)$ time a shortest non-contractible arc with endpoints in $\delta$.  This arc has multiplicity at most 2 and is edge-disjoint from $\delta$.

\item[\textup{(d)}]
If $\Sigma$ has exactly two boundaries, we can find in $O(n\log n)$ time a shortest arc with endpoints in both boundaries.  This arc has multiplicity 1 and is edge-disjoint from the boundary of $\Sigma$.
\end{itemize}
\end{lemma}

\begin{proof}
    \begin{itemize}
        \item[(a)] See Erickson and Har-Peled~\cite[Lemma 5.2]{schema}.
        \item[(b)] See Cabello \etal~\cite{cdem-fotc-10}.
        \item[(c)] Contract $\delta$ to a point $p$ and construct a shortest non-contractible
            loop with basepoint $p$ as in part (a). This is the desired arc in $\Sigma$.
        \item[(d)] Select boundaries $\delta$ and $\delta'$ of $\Sigma$, contract them to points $p$ and $p'$,
            and construct in $O(n\log n)$ time the shortest path from $p$ to $p'$.
        \end{itemize}
\end{proof}

\begin{lemma}\label{le:killboundary}
    Let $\Sigma$ be a combinatorial surface with complexity $n$, genus $g$, and $b$ boundaries.
    Let $\Sigma'$ be the surface obtained by pasting a disk to each boundary of $\Sigma$.
    We can find a shortest non-contractible cycle in $\Sigma$ in $O(b n\log n)$ time plus
    the time used to find a shortest non-contractible cycle in $\Sigma'$.
\end{lemma}
\begin{proof}

Number the boundaries $\delta_1,\delta_2,\dots,\delta_b$, and set
$\Sigma_0=\Sigma$. For $i \ge 1$, let $\Sigma_i$ be the surface obtained from
$\Sigma_{i-1}$ by pasting a disk to $\delta_i$.  In particular,
$\Sigma_b=\Sigma'$.  For each $i$, let $C_i$ be a shortest cycle homotopic
to $\delta_i$ in $\Sigma_{i-1}$. We can compute each cycle $C_i$ in $O(n
\log n)$ time.

Lemma~\ref{le:gluedisk} implies that a shortest non-contractible
cycle in $\Sigma_{i-1}$ is either $C_i$ or a shortest non-contractible
cycle in $\Sigma_{i-1}$.  Thus, the shortest non-contractible cycle in
$\Sigma$ is either the shortest among $C_1,\dots, C_b$ or a shortest
non-contractible cycle in $\Sigma'$.
\end{proof}

\begin{theorem}\label{th:non-con}
    Let $\Sigma$ be a combinatorial surface with complexity $n$, genus $g$, and $b$ boundaries, with generic edge weights.
    We can find a shortest non-contractible cycle of $\Sigma$ in $O((g^2+b) n\log n)$ time.
\end{theorem}
\begin{proof}

We apply Lemma~\ref{le:killboundary} to $\Sigma$. We spend $O(bn\log n)$
time, and have reduced the problem to find a shortest
non-contractible cycle in a surface $\Sigma'$ of genus $g$ and no
boundary.  We next give an iterative algorithm that reduces either
the genus or the number of boundaries of the subproblems in each
iteration.  Through the algorithm we use a variable $\emph{minCycle}$ 
to store the shortest non-contractible cycle found so far.
In each iteration we find one new non-contractible cycle $\gamma$ whose length
is compared against the length of $\emph{minCycle}$, 
and $\emph{minCycle}$ is updated if needed. 
The algorithm stops when each remaining component is a
topological disk.  

We distinguish three cases. $\Sigma'$ denotes the
surface in the subproblem. Through the algorithm the surface $\Sigma'$ has
genus at most $g$ and each component of $\Sigma'$ has at most two boundaries.
    \begin{itemize}
        \item[(a)] If $\Sigma'$ is a surface without boundary, we choose a point $x\in \Sigma'$
            and find a shortest non-contractible loop $\ell_x$ through $x$ (Lemma~\ref{le:find}(a)).
            Lemma~\ref{le:crossings} implies that
            there is a shortest non-contractible cycle in $\Sigma'$
            crossing $\ell_x$ at most once.
            We compute the shortest cycle $\gamma$ in $\mathit{Cross}(\ell_x)$ using Lemma~\ref{le:arc} 
			and update $\emph{minCycle}$ with $\gamma$, if it is shorter.
            Then we set $\Sigma' = \Sigma' \setminus \ell_x$. Note that if $\ell_x$ is
            separating, then $\Sigma'$ has two connected components and $\mathit{Cross}(\ell_x)=\emptyset$.
            If $\Sigma'$ has complexity $m$, we spend $O(g m \log m)$ time in this iteration.
            
        \item[(b)] If $\Sigma'$ has a component with non-zero genus and exactly one boundary $\delta$, we find a shortest non-contractible
            arc $\alpha$ with endpoints in $\delta$ (Lemma~\ref{le:find}(c)).
            Lemma~\ref{le:crossings2} implies that
            there is a shortest non-contractible cycle in $\Sigma'$ that either
            crosses $\alpha$ at most once or is homotopic to $\delta$.
            We compute the shortest cycle in $\mathit{Cross}(\alpha)$ and the shortest cycle homotopic to
            $\delta$ in $O(g n \log n)$ time (Lemmas~\ref{le:arc} and~\ref{le:find}(b)).  
			The shortest of these two cycles is compared against $\emph{minCycle}$, 
			and $\emph{minCycle}$ is updated if necessary.
			We then set $\Sigma' = \Sigma' \setminus \alpha$.
            If $\Sigma'$ has complexity $m$, we spend $O(g m \log m)$ time in this iteration.

        \item[(c)] If $\Sigma'$ has one component with exactly two boundaries, 
			we find a shortest arc $\alpha$ connecting them (Lemma~\ref{le:find}(d)).  
			Lemma~\ref{le:crossings3} implies that
            there is a shortest non-contractible cycle $\ell$ in $\Sigma'$
            crossing $\alpha$ at most once.
            We compute the shortest cycle $\gamma$ in $\mathit{Cross}(\alpha)$ using Lemma~\ref{le:arc},
            and update $\emph{minCycle}$ with $\gamma$, if it is shorter.
            Then we replace $\Sigma'$ with $\Sigma' \setminus \alpha$.
            If $\Sigma'$ has complexity $m$, we spend $O(g m \log m)$ time in this iteration.
    \end{itemize}

This finishes the description of the algorithm. The algorithm has
$O(g)$ iterations; starting with no boundary, we iterate in case
(a) once and then cases (b) or (c) at most $2g$ times.  The arcs and
loops that are used to cut, which are obtained from Lemma~\ref{le:find},
have multiplicity at most two and are edge-disjoint from the
boundary.  Thus any subsurface $\Sigma'$ considered in a subproblem has
at most four copies of an edge of $\Sigma$.  This means that any surface $\Sigma'$ 
has complexity $O(n)$ and in each iteration we spend $O(g n\log n)$ time.
\end{proof}

\def\tslack{\widetilde{\smash{\slack}}}
\def\tdist{\widetilde{\smash{\dist}}}
\def\tw{\widetilde{w}}

\section{Enforcing Genericity}
\label{S:perturb}

So far we have assumed that input edge weights are \emph{generic}.  Specifically, our analysis requires that all vertex-to-vertex shortest paths are unique and that exactly one dart becomes tense at each critical value of $\lambda$.  This assumption can be enforced---or from the algorithm user's viewpoint, removed---using standard perturbation techniques.  We describe two different perturbation schemes in this section.  The first is a simple randomized method based on the Isolation Lemma of Mulmuley \etal~\cite{mvv-memi-87} that succeeds with high probability, without increasing the asymptotic running time of our algorithms.  The second is a more complex deterministic method based on lexicographic perturbation that increases our worst-case running times by a factor of $O(\log n)$.

Klein's planar multiple-source shortest-path algorithm \cite{k-msspp-05} and its recent improvement for unweighted graphs by Eisenstat and Klein \cite{ek-lafms-13} avoid degeneracy issues by pivoting the \emph{leafmost} tense dart into at every iteration, thereby maintaining the \emph{rightmost} shortest-path tree.  However, it is not clear how to extend their strategy to graphs on higher-genus surfaces, because the dual complement of a shortest-path tree is no longer a tree.

\subsection{Random Perturbation}

Our first perturbation scheme simply adds a small amount of random noise to the edge weights.  Intuitively, we modify the weight of each dart by defining
\[
	\tw(\dart{x}{y}) := w(\dart{x}{y}) + \omega(\dart{x}{y})\cdot\e
\]
for some perturbation value $\omega(\dart{x}{y})$, chosen independently and uniformly at random from the set $\set{1, 2, \dots, n^4}$, and some sufficiently small $\e>0$.  Instead of choosing an explicit numerical value $\e$, we consider the limiting behavior as $\e$ approaches zero.  Equivalently, we define the perturbed weight of each dart as a two-dimensional vector
\[
	\tw(\dart{x}{y}) := (w(\dart{x}{y}),~\omega(\dart{x}{y})).
\]
Shortest-path lengths and slacks are now also vectors.  We compare these vectors lexicographically; that is, $(d, \delta) > (d', \delta')$ if and only if either $d>d'$, or $d=d'$ and $\delta > \delta'$.  Thus, if the given edge weights happen to be generic, the perturbation parameters will never be used.  Because each perturbation value is a $O(\log n)$-bit integer, we can add, subtract, and compare length vectors in constant time. 

To prove that that this perturbation scheme yields generic edge weights with high probability, we use the following mild generalization of the Isolation Lemma of Mulmuley \etal~\cite[Lemma 1]{mvv-memi-87}; see also Spencer \cite[Theorem 3.1]{s-pm-95}.

\begin{lemma} 
\label{L:isolation}
Let $\omega = (\omega_1, \omega_2, \dots, \omega_n)$ be a vector of $n$ integers,  chosen uniformly at random from $\set{1,2,\dots,N}^n$, and let $\Pi$ be an arbitrary set of vectors in $\set{-1,0,1}^n$.  With probability at least $1-3n/N$, there is a unique vector $\pi\in \Pi$ that minimizes the inner product $\seq{\omega,\pi}$.
\end{lemma}

\begin{proof}
Fix an index $i$.  For each $\e\in\set{-1,0,1}$, let $\Pi_{i,\e}$ denote the subset of vectors in~$\Pi$ whose $i$th coefficient is $\e$, and define
\[
	m_{i,\e} :=
	\Min_{\pi\in\Pi_{i,\e}}\, \seq{\omega,\pi} - \e\omega_i.
\]
(We assume here that every set $\Pi_{i,\e}$ is non-empty.)  Observe that $m_{i,\e}$ is independent of $\omega_i$ because the terms $\e\omega_i$ cancel out.  Let $M_i = \set{m_{i,\e} + \e\omega_i \mid \e\in\set{-1,0,1}}$, and observe that $\min_{\pi\in\Pi} \seq{\omega,\pi}$ is the smallest element of $M_i$.  Call the coefficient $\omega_i$ \emph{singular} if $M_i$ has fewer than three distinct elements, or equivalently, if at least one of the following equations holds:
\[
	\omega_i = m_{i,-1} - m_{i,0},
	\qquad
	\omega_i = m_{i,0} - m_{i,1},
	\qquad
	\omega_i = (m_{i,-1} - m_{i,1})/2.
\]
Because $m_{i,\e} - m_{i,\e'}$ is independent of $\omega_i$, for any $\e,\e'\in \set{-1,0,1}$, and $\omega_i$ is drawn uniformly at random from a set of size $N$, the probability that $\omega_i$ is singular is at most $3/N$.  (In fact, if any set $\Pi_{i,\e}$ is empty, this probability is at most $1/N$.)  It follows that with probability at least $1-3n/N$, \emph{none} of the coefficients of $\omega$ is singular.

Call a vector $\pi\in\Pi$ \emph{optimal} if it minimizes the inner product $\seq{\omega,\pi}$.  If there are two optimal vectors, they differ in some index $i$, and the corresponding coefficient $\omega_i$ is singular.
\end{proof}


In our applications of Lemma \ref{L:isolation}, $n$ is the number of darts in the input graph, $\omega$ is the vector of perturbation values (indexed by darts), $N = n^4$, and each vector in $\Pi$ is either the indicator vector of a simple directed path in $G$ or the difference between two such vectors.

\begin{lemma}
With probability at least $1-O(1/n)$, all vertex-to-vertex shortest paths are unique.
\end{lemma}

\begin{proof}
Fix two vertices $x$ and $y$, and let $\Pi$ be the set of all indicator vectors for paths from $x$ to $y$ that are shortest with respect to the \emph{unperturbed} edge weights.  Lemma \ref{L:isolation} implies that the path in $\Pi$ with minimum \emph{perturbed} length is unique with probability at least $1-3/n^3$.  There are $O(n^2)$ pairs of vertices.
\end{proof}

\begin{lemma}
With probability at least $1-O(1/n)$, the tensest active dart is always unique.
\end{lemma}

\begin{proof}
Fix a point in our parametric shortest-path algorithm's execution between pivots, and let $\dart{u}{v}$ be the dart containing the moving source vertex.  For any other dart $\dart{x}{y}$, define
\[
	\excess(\dart{u}{v},\dart{x}{y})
	:= \dist(v,x) + w(\dart{x}{y}) - \dist(u,y),
\]
where $\dist(v,x)$ denotes the shortest-path distance from $v$ to $x$.  If $\dart{x}{y}$ is active, then 
\[
	\slack_\lambda(\dart{x}{y})
	~=~
	\excess(\dart{u}{v},\dart{x}{y}) ~+~ (1-\lambda) w(\dart{u}{v}) - \lambda w(\dart{v}{u}).
\]
It follows that the active dart with minimum slack is also the active dart with minimum excess.

For any active dart $\dart{x}{y}$, let $\pi(\arc{u}{v}, \arc{x}{y})$ denote the vector in $\set{-1,0,1}$ where each $1$ indicates either a dart on the shortest path from $v$ to $x$ or the dart $\dart{x}{y}$, and each $-1$ indicates a dart on the shortest path from $u$ to $y$.  (Because $x$ is blue and $y$ is red, these two shortest paths are disjoint.)  Let~$A_{\min}$ denote the set of active darts with minimum \emph{unperturbed} excess, and let $\Pi$ denote the set of all vectors $\pi(\arc{u}{v}, \arc{x}{y})$ such that $\dart{x}{y}\in A_{\min}$.  Lemma \ref{L:isolation} now implies that the dart in $A_{\min}$ with \emph{perturbed} excess is unique with probability at least $1-3/n^3$.

Finally, the set of active darts changes at most $O(n)$ times as we move the source point along any dart $\dart{u}{v}$, and there are $n$ possible source darts $\dart{u}{v}$.  It follows that there are at most $O(n^2)$ distinct sets $A_{\min}$ during the entire execution of the algorithm, even if we do not restrict the source dart to lie on the boundary of a particular face.
\end{proof}

We can make the collision probability as small as we like by choosing the perturbation values from a larger set of integers.  Specifically, to reduce the probability of a collision to $O(1/n^c)$, it suffices to draw perturbation values  from a set of size $\Theta(n^{c+3})$.

\begin{corollary}
Let $G$ be an directed graph with $n$ vertices and arbitrary edge weights, cellularly embedded in a surface of genus~$g$, and let $f$ be any face of $G$.  In $O(g n\log n)$ time and space \textbf{with high probability}, we can construct a data structure that supports the following operations:
\begin{itemize}\cramped
\item
Given any vertex $u$ on the boundary of $f$ and any other vertex $v$, return the shortest-path distance from $u$ to $v$ in $O(\log n)$ time.
\item
Given any vertex $u$ on the boundary of $f$ and any other vertex $v$, return the shortest path from $u$ to $v$ in $O(\log n + k \log \Delta)$ time, where $k$ is the number of edges in the path and $\Delta$ is the maximum degree of $G$. 
\end{itemize}
\end{corollary}

\begin{corollary}
Let $\Sigma$ be a combinatorial surface with complexity $n$ and genus $g$, with arbitrary edge weights.  We can find a shortest non-separating cycle of $\Sigma$ in $O(g^2 n\log n)$ time \textbf{with high probability}.
\end{corollary}

\begin{corollary}
Let $\Sigma$ be a combinatorial surface with complexity $n$, genus $g$, and $b$ boundaries,  with arbitrary edge weights.  We can find a shortest non-contractible cycle of $\Sigma$ in $O((g^2+b) n\log n)$ time \textbf{with high probability}.
\end{corollary}

\subsection{Lexicographic Perturbation}

Our second method is an efficient implementation of the \emph{lexicographic perturbation}, first proposed by Charnes \cite{c-odlp-52} and Dantzig~\etal~\cite{dow-gsmml-55}.  Hartvigsen and Mardon \cite{hm-apmcp-94} proposed a similar perturbation scheme for efficiently computing minimum-cost cycle bases in planar graphs; their algorithm (including the perturbation scheme) was more recently improved by Borradaile \etal~\cite{bsw-mscop-10, w-mcbap-09}.\footnote{Hartvigsen and Mardon’s perturbation scheme breaks ties by comparing indices of vertices instead of indices of darts.  Our implementation can be adapted to their scheme, with the same running time, with only minor modifications.}

Arbitrarily index the darts in $G$ from $1$ to $n$, and let $i(\dart{x}{y})$ denote the index of $\dart{x}{y}$.  We intuitively define the perturbed weight of $\dart{x}{y}$ as 
\[
	\tw(\dart{x}{y}) := w(\dart{x}{y}) + \e^{i(\dart{x}{y})},
\]
for some sufficiently small $\e>0$.  Again, instead of choosing an explicit numerical value for $\e$, we consider the limiting behavior as $\e$ approaches zero.  Equivalently, we define the perturbed weight of $\dart{x}{y}$ to be an $(n+1)$-dimensional vector, whose first coefficient is the unperturbed weight, whose $(i(\dart{x}{y})+1)$th coordinate is $1$, and whose other coordinates are all zero.  As in the previous scheme, path lengths and slacks are also vectors, which we compare lexicographically; the result of a comparison is determined by the first coordinate in which the two vectors differ.

Two linear combinations of lexicographically perturbed edge weights are equal if and only if they have identical coefficient vectors.  Thus, any two simple paths have distinct lengths, so shortest paths are unique, and any two dart pairs have distinct excesses, so the tensest active dart is always unique.  However, a naïve implementation of lexicographic perturbation requires $O(n)$ time for each addition, subtraction, or comparison of length vectors.

To implement lexicographic perturbation more efficiently, we store the shortest-path tree $T$ in a dynamic tree data structure that stores the \emph{indices} of all darts in $T$ and supports the following operations:
\begin{itemize}
\item
$\textsc{Create}(v)$: Create a new one-vertex tree.
\item
$\textsc{Link}(u, v)$: Add the edge $uv$ to $T$.
\item
$\textsc{Cut}(e)$: Remove edge $e$ from $T$.
\item
$\textsc{LCA}(u, v)$: Return the least-common ancestor of nodes $x$ and $y$.
\item
$\textsc{PathMinIndex}(u,v)$: Return the minimum-index dart on the directed path in $T$ from $x$ to $y$.
\end{itemize}
This new \emph{minimum-index tree structure} can be implemented using Euler-tour trees \cite{hk-rfdga-99, tarjan-eulertrees} or self-adjusting top trees~\cite{tw-satt-05}, so that each operation is supported in $O(\log n)$ amortized time.  We emphasize that the minimum-index structure is distinct from the primal tree structure that our algorithm already maintains, as described in Section~\ref{SS:dyn-trees}.

Our algorithms make only two types of comparisons that depend on the edge weights: (1)~comparisons between path-lengths in Dijkstra's algorithm when constructing the initial shortest-path tree, and (2) comparisons between slacks of active darts during the main algorithm.  Moreover, the running time of our algorithm is dominated by the total time to perform these comparisons.  Each of these comparisons can be performed in $O(\log n)$ time using the minimum-index structure, as described below.  We emphasize that our perturbation scheme does not support comparisons between \emph{arbitrary} path lengths or \emph{arbitrary} slacks, but only those comparisons necessary to support our algorithm.

\paragraph{Dijkstra's algorithm.}

We construct the initial shortest-path tree using Dijkstra's classical algorithm, which can be formulated as follows.  The algorithm maintains a tree $T$ of shortest paths from the source vertex $s$ to a subset of the vertices, initially containing only $s$.  Each vertex $x$ maintains a tentative distance $\dist(x)$ (initially $\infty$) and a tentative predecessor $\pred(x) \in T$ (initially undefined).  At each iteration, the algorithm finds the vertex $x\not\in T$ with  with minimum tentative distance, adds the dart $\dart{\pred(x)}{x}$ to~$T$, and updates the distances and predecessors of the out-neighbors of $x$.  

All comparisons made by this formulation of Dijkstra's algorithm have the following form: Given two darts $\dart{x}{y}$ and $\dart{x'}{y'}$, where $x$ and $x'$ are in $T$ (and $y$ and $y'$ are not), decide whether $\dist(x) + w(\dart{x}{y})$ is smaller or larger than $\dist(x') + w(\dart{x'}{y'})$.  (It is possible that $x=x'$ or $y=y'$, but not both.)  Moreover, the running time of the algorithm is dominated by the time of these comparisons.  Let $\hat{x}$ denote the least common ancestor of $x$ and $x'$.  The result of this comparison depends on the sign of the following expression:
\[
	\dist(\hat{x},x) + w(\dart{x}{y}) - \dist(\hat{x}, x') - w(\dart{x'}{y'}).
\]
If the two \emph{unperturbed} path lengths are equal, the sign of the \emph{perturbed} expression depends on the minimum index among $\dart{x}{y}$, $\dart{x'}{y'}$, and the darts in the paths in $T$ from $\hat{x}$ to $x$ and $x'$.

\begin{algo}
	\textul{$\textsc{SmallerLexDistance}(\dart{x}{y}, \dart{x'}{y'})$}\+
\\	if $\dist(x) + w(\dart{x}{y}) < \dist(x') + w(\dart{x'}{y'})$\+
\\		return $\dart{x}{y}$\-
\\	if $\dist(x) + w(\dart{x}{y}) > \dist(x') + w(\dart{x'}{y'})$\+
\\		return $\dart{x'}{y'}$\-
\\[0.5ex]
	$\hat{x} \gets \textsc{LCA}(x, x')$
\\	$\emph{min}^+ \gets
		\min\Set{\textsc{PathMinIndex}(\hat{x},x),~
			\;i(\dart{x}{y})}$
\\	$\emph{min}^- \gets
		\min\Set{\textsc{PathMinIndex}(\hat{x},x'),~
			i(\dart{x'}{y'})}$
\\[0.5ex]
	if $\emph{min}^+ < \emph{min}^-$\+
\\		return $\dart{x}{y}$\-
\\	else\+
\\		return $\dart{x'}{y'}$\-
\end{algo}
\textsc{SmallerLexDistance} uses a constant number of dynamic forest operations and therefore runs in $O(\log n)$ amortized time.  Replacing all simple comparisons in Dijkstra’s algorithm with \textsc{SmallerLexDistance} increases its running time by a factor of $O(\log n)$; the resulting algorithm computes a tree of lexicographically-shortest paths.  Essentially the same algorithm was previously described by Borradaile \etal\ \cite[Section 8.1]{bsw-mscop-10x}.

\paragraph{Main algorithm.}

In our main algorithm, we store only the \emph{unperturbed} shortest-path distances and slacks in the primal tree and dual forest structures, respectively; the lexicographically perturbed weights are used only to break ties in \textsc{MinPath} queries to the dual forest structure and in maintaining the priority queue of green subtrees, as described in Section~\ref{SS:fast-pivot}.  In both cases, we need to break ties between active darts with the same unperturbed slack.

Suppose we are moving the source vertex along $\dart{u}{v}$.  As in our randomized perturbation scheme, the key observation is that the \emph{difference} between slacks is independent of the parameter $\lambda$.  For any two active darts $\dart{x}{y}$ and $\dart{x'}{y'}$, we have
\begin{align*}
	\lefteqn{\slack(\dart{x}{y}) - \slack(\dart{x'}{y'})}\qquad\qquad
\\	&=
	\excess(\dart{u}{v}, \dart{x}{y}) - \excess(\dart{u}{v}, \dart{x'}{y'})
\\	&=
	\dist(v,x) - \dist(v,x') + \dist(u,y') - \dist(u,y)
	+ w(\dart{x}{y}) - w(\dart{x'}{y'}).
\end{align*}
Let $\hat{x}$ be the least common ancestor of $x$ and $x'$ in the current shortest-path tree $T_\lambda$, and let $\hat{y}$ be the least common ancestor of $y$ and $y'$ in $T_\lambda$.  Then we have 
\begin{align*}
	\lefteqn{\slack(\dart{x}{y}) - \slack(\dart{x'}{y'})}\qquad\qquad
\\	&=
	\dist(\hat{x},x) - \dist(\hat{x},x')
	+ \dist(\hat{y},y') - \dist(\hat{y},y)
	+ w(\dart{x}{y}) - w(\dart{x'}{y'}).
\end{align*}
The paths from $\hat{x}$ to $x$ and $x'$ lie in the blue subtree of $T_\lambda$; the paths from $\hat{y}$ to $y$ and $y'$ lie in the red subtree of $T_\lambda$; and the darts $\dart{x}{y}$ and $\dart{x'}{y'}$ are both green.  Thus, all six of these directed paths are edge-disjoint.  If the \emph{unperturbed} slacks of $\dart{x}{y}$ and $\dart{x'}{y'}$ are equal, then the result of the \emph{perturbed} slack comparison depends on the minimum index among all darts in these six paths.
\begin{algo}
	\textul{$\textsc{SmallerLexSlack}(\dart{x}{y}, \dart{x'}{y'})$}\+
\\	if $\slack(\dart{x}{y}) < \slack(\dart{x'}{y'})$\+
\\		return $\dart{x}{y}$\-
\\	if $\slack(\dart{x}{y}) > \slack(\dart{x'}{y'})$\+
\\		return $\dart{x'}{y'}$\-
\\[0.5ex]
	$\hat{x} \gets \textsc{LCA}(x, x')$
\\	$\hat{y} \gets \textsc{LCA}(y, y')$
\\	$\emph{min}^+ \gets
		\min\Set{\textsc{PathMinIndex}(\hat{x},x),~
			\;\textsc{PathMinIndex}(\hat{y},y'),~
			i(\dart{x}{y})}$
\\	$\emph{min}^- \gets
		\min\Set{\textsc{PathMinIndex}(\hat{x},x'),~
			\textsc{PathMinIndex}(\hat{y},y),~
			\;i(\dart{x'}{y'})}$
\\[0.5ex]
	if $\emph{min}^+ < \emph{min}^-$\+
\\		return $\dart{x}{y}$\-
\\	else\+
\\		return $\dart{x'}{y'}$\-
\end{algo}

\textsc{SmallerLexSlack} uses a constant number of dynamic forest operations and therefore runs in $O(\log n)$ amortized time.  Replacing simple slack comparisons with \textsc{SmallerLexSlack} increases the amortized time for each dual forest operation (\textsc{MinPath}, \textsc{AddPath}, \textsc{Cut}, or \textsc{Link}) from $O(\log n)$ to $O(\log^2 n)$.

\begin{corollary}
Let $G$ be an directed graph with $n$ vertices and arbitrary edge weights, cellularly embedded in a surface of genus~$g$, and let $f$ be any face of $G$.  In $O(g n\log^2 n)$ time and $O(gn\log n)$ space \textbf{in the worst case}, we can construct a data structure that supports the following operations:
\begin{itemize}\cramped
\item
Given any vertex $u$ on the boundary of $f$ and any other vertex $v$, return the shortest-path distance from $u$ to $v$ in $O(\log n)$ time.
\item
Given any vertex $u$ on the boundary of $f$ and any other vertex $v$, return the shortest path from $u$ to $v$ in $O(\log n + k \log \Delta)$ time, where $k$ is the number of edges in the path and $\Delta$ is the maximum degree of $G$. 
\end{itemize}
\end{corollary}

\begin{corollary}
Let $\Sigma$ be a combinatorial surface with complexity $n$ and genus $g$, with arbitrary edge weights.  We can find a shortest non-separating cycle of $\Sigma$ in $O(g^2 n\log^2 n)$ time \textbf{in the worst case}.
\end{corollary}

\begin{corollary}
Let $\Sigma$ be a combinatorial surface with complexity $n$, genus $g$, and $b$ boundaries, with arbitrary edge weights.  We can find a shortest non-contractible cycle of $\Sigma$ in $O((g^2+b) n\log^2 n)$ time \textbf{in the worst case}.
\end{corollary}

\newpage

\bibliographystyle{plain}

\bibliography{biblio}

\end{document}